\documentclass{amsart}

\usepackage{amsmath,amsfonts,amssymb,amsthm,amscd,comment,euscript}
\usepackage{mathbbol, stmaryrd}

\usepackage[all]{xy}
\usepackage{graphicx}
\usepackage{enumerate} 
\usepackage[colorlinks=true]{hyperref} 
\usepackage[usenames,dvipsnames]{xcolor}

\usepackage{hyperref} 

\theoremstyle{plain}
\newtheorem{lem}{Lemma}[section]
\newtheorem{cor}[lem]{Corollary}
\newtheorem{prop}[lem]{Proposition}
\newtheorem{thm}[lem]{Theorem}

\theoremstyle{definition}
\newtheorem{ex}[lem]{Example}
\newtheorem{rem}[lem]{Remark}
\newtheorem{dfn}[lem]{Definition}

\newtheorem{conj}[lem]{Conjecture}

\newcommand{\dis}{\displaystyle}

\newcommand{\al}{\alpha}

\newcommand{\vac}{\mathbf 1}

\title[Some remarks on associated varieties]{Some remarks on associated varieties of vertex operator superalgebras}
\author{Hao Li}
\address{Department of Mathematics and Statistics, SUNY-Albany, Albany 12222, NY, USA}
\email{hli29@albany.edu}
\begin{document}
\begin{abstract} We study several families of vertex operator superalgebras from a jet (super)scheme point of view.
We provide new examples of vertex algebras which are "chiralizations" of their
Zhu's Poisson algebras $R_V$. Our examples come from affine $C_\ell^{(1)}$-series vertex algebras ($\ell \geq 1$), certain $N=1$ superconformal vertex algebras, Feigin-Stoyanovsky principal subspaces, Feigin-Stoyanovsky type subspaces, graph vertex algebras $W_{\Gamma}$, and extended Virasoro vertex algebra. We also give a counterexample to the chiralization property for the $N=2$ superconformal vertex algebra of central charge $1$.
\end{abstract}
\maketitle
\section{Introduction}






 Beilinson, Feigin and Mazur \cite{beilinson1991introduction} first introduced the notions of singular support and lisse representation in order to study Virasoro (vertex) algebra. Arakawa later extended these notions to any finitely strongly generated, non-negatively graded vertex algebra $V$. More precisely, via a canonical decreasing filtration $\left\{F_{p}(V)\right\}$ introduced in \cite{li2005abelianizing}, we can associate to $V$  a positively graded vertex Possion vertex algebra $gr^{F}(V)$. The spectrum of  $gr^{F}(V)$ is called singular support of $V$ and is denoted by $SS(V)$. With respect to this filtration, $V/F_{1}(V)$ is the Zhu $C_{2}$-algebra $R_V$. The reduced spectrum $X_{V}={\rm Specm}(R_{V})$ is a Poisson variety which is called the associated variety of $V.$ A large body of work has been devoted to
descriptions of associated variety for various  vertex operator algebras
 \cite {arakawa2012remark,arakawa2018joseph,arakawa2018irreducibility,arakawa2017sheets}.
Certainly the most prominent examples from this point of view are well-known lisse, or $C_2$-cofinite vertex algebras characterized by ${\rm dim}(X_{V})=0$. Arakawa and Kawasetsu relaxed this condition to quasi-lisse in \cite{arakawa2018quasi} which requires that $X_{V}$ has finitely many symplectic leaves. Associated varieties are important in the geometry of Higgs branches in 4d/2d dualities in physics \cite{beem2015infinite}.

According to \cite[Proposition 2.5.1]{arakawa2012remark}, the embedding \[R_{V}\hookrightarrow gr^{F}(V)\] can be extended to a 
surjective homomorphism of vertex Poisson algebras
 \[\psi:J_{\infty}(R_{V})\twoheadrightarrow gr^{F}(V)\] where $J_{\infty}(R_{V})$ is the (infinite) jet algebra of $R_{V}.$ The map $\psi$ induces an injection from the singular support
into the jet scheme of the associated scheme of $V$, $\widetilde{X}_{V}={\rm Spec}(R_{V})$,
\[\phi:SS(V)\hookrightarrow J_{\infty}(\widetilde{X}_{V}).\]In \cite{arakawa2018arc}, authors showed that $\phi$ is an isomorphism as varieties if $V$ is quasi-lisse. It was shown in \cite{van2018chiral} that if the map $\psi$ is an isomorphism, then one can compute Hochschild homology of the Zhu algebra via the chiral homology of elliptic curves. Proving that $\psi$ is an isomorphsim or finding the kernel of $\psi$ turns out to be subtle. In \cite{arakawa2018singular}, authors provided several examples for which $\psi$ is not an isomorphism, including the $\mathbb{Z}_2$-orbifold of the rank one Heisenberg algebra. Finding the kernel of $\psi$ even for this example seems quite hard (see also \cite{van2020singular}).

For a vertex algebra $V$ where $\psi$ is an isomorphism we obtain a very interesting (and important) consequence
$${\rm ch}[V](\tau)=HS_q(J_\infty(R_V)),$$
where the left-hand side is the graded dimension of $V$ and the right-hand side is the Hilbert series. The left-hand side has often combinatorial interpretations which
in turn can provide a non-trivial information about the jet scheme.

This work is our modest attempt to try to generalize above notions to vertex superalgebra case. We first generalize the notion of  associated variety to vertex superalgebras. Then we investigate the map $\psi$ in the cases of affine vertex algebras, rank one lattice vertex superalgebras including the simple $N=2$ superconformal vertex algebra at level 1, Feigin-Stoyanovsky principal subspaces, Feigin-Stoyanovsky type subspaces, simple $N=1$ vertex algebra associated with $(2,4k)-$minimal model and certain extended Virasoro vertex algebras. Along the way, we get some interesting character fomulas and the bases of vertex algebra. We provide an example which is simple $N=2$ vertex algebra at level 1, where $\psi$ is not an isomorphism. Moreover we make a conjecture about its kernel.
We end the paper with a brief glimpse at our plans for future research.



\section{Definitions and Preliminary results}

\begin{dfn}

Let $V$ be a superspace, i.e., a $\mathbb{Z}_{2}-$graded vector space. $V=V_{\overline{0}} \oplus V_{\overline{1}}$ where $\{ {\overline{0}}, \overline{1} \}=\mathbb{Z}_2$. If $a\in V_{p(a)}$, we say that the element $a$ has parity $p(a)\in \mathbb{Z}_2$.

A field is a formal series of the form $a(z)=\sum_{n\in \mathbb{Z}}a_{(n)}z^{-n-1}$ where $a_{n}\in $ End$(V)$ and for each $v\in V$ one has $$a_{(n)}v =0$$ for $n\gg 0$.

We say that a field $a(z)$ has parity $p(a)\in \mathbb{Z}_2$ if $$a_{(n)}V_{\alpha}\in V_{{\alpha}+p(a)}$$ for all $\alpha\in \mathbb{Z}_2$, $n\in \mathbb{Z}.$

A vertex superalgebra contains the following data: a vector space of states $V$, the vacuum vector $\mathbf{1} \in V_{\overline{0}},$ derivation $T$, and the state-field correspondence map $$a\longmapsto Y(a,z)=\sum_{n\in \mathbb{Z}} a_{(n)}z^{-n-1},$$ satisfying the following axioms:
\begin{itemize}
    \item \textbf{(translation coinvariance):} $[T,Y(a,z)]=\partial Y(a,z)$.
    \item \textbf{(vacuum):} $Y(\mathbf{1} ,z)=Id_{V}$, $Y(a,z)\mathbf{1} |_{z=0}=a,$
    \item \textbf{(locality):} $(z-w)^{N}Y(a,z)Y(b,w)=(-1)^{p(a)p(b)}(z-w)^{N}Y(b,w)Y(a,z)$
    for $N\gg 0$.
\end{itemize}

\end{dfn}

In particular, a vertex algebra $V$ is called supercommutative if $a_{(n)}=0$ for $n\geq 0$. It is well-known that the category of  commutative vertex superalgebras is equivalent with the category of unital  commutative associative superalgebra equipped with an even derivation.

We say a vertex algebra ${V}$ is generated by a subset $\mathcal{U} \subset {V}$ if any element of ${V}$ can be written as a finite linear combination of terms of the form

\begin{align*}
b^{1}_{(i_1)} {b^{2}}_{(i_2)} \ldots  {b^{n}}_{(i_n)} \mathbf{1}
\end{align*}

\noindent for $b^k \in \mathcal{U}$, $i_k \in \mathbb{Z}$, and $n \geq 0$.  If every element of $V$ can be written with $i_k <0$, we write ${V}=\langle \mathcal{U} \rangle_S$ and say $V$ is strongly generated by $\mathcal{U}$.

\begin{ex}(see for instance \cite{zheng2017vertex}) \label{aflie}
  Let $\mathfrak{g}$ be a finite dimensional Lie superalgebra with a nondegenerate even supersymmetric invariant bilinear form $(\cdot,\cdot) $. We can associate the affine Lie superalgebra $\widehat{\mathfrak{g}}$  to the pair $(\mathfrak{g},(\cdot,\cdot))$. Its universal vacuum representation of level $k$, $V_{\widehat{\mathfrak{g}}}(k,0)$, is a vertex superalgebra.
  In particular, when $\mathfrak{g}$ is a simple Lie algebra, $V_{\widehat{\mathfrak{g}}}(k,0)$ has an unique maximal ideal $I_{\widehat{\mathfrak{g}}}(k,0).$
And $L_{\widehat{\mathfrak{g}}}(k,0)=V_{\widehat{\mathfrak{g}}}(k,0)/I_{\widehat{\mathfrak{g}}}(k,0)$ is also a vertex algebra.
\end{ex}

\begin{ex}\cite{kac1998vertex}\label{fermion} 
To any $n$ dimensional superspace $A$ with a non-degenerate anti-supersymmetric bilinear form $(\cdot,\cdot)$, we can associate a Lie superalgebra $C_{A}$. If we fix a basis of $A$: \[\left\{\phi^{1},\ldots,\phi^{n}\right\},\] the free fermionic vertex algebra $\mathcal{F}$ associated to $A$, is a vertex superalgebra strongly generated by $\phi^{i}_{(-\frac{1}{2})}\mathbf{1}$ $(1\leq i\leq n)$ where $Y(\phi^{i}_{(-\frac{1}{2})}\mathbf{1},z)=\displaystyle \sum_{n\in \frac{1}{2}+\mathbb{Z}}\phi_{(n)}^i z^{-n-\frac{1}{2}}.$
\end{ex}

\begin{dfn}

A vertex superalgebra $V$ is called a vertex operator superalgebra if it is $\frac12 \mathbb{Z}$-graded,
$$V=\coprod_{n \in \frac12 \mathbb{Z}} V_{(m)},$$
with a conformal vector $\omega$ such that the set of operators $\left\{L_{(n)}, id_{V}\right\}_{n\in \mathbb{Z}}$ with 
$L_{(n)}=\omega_{(n+1)}$ defines a representation of Virasoro algebra on $V$; that is $$[L_{(n)},L_{(m)}]=(m-n)L_{(m+n)}+\frac{m^{3}-m}{12}\delta_{m+n,0}c_{V}$$ for $m,n\in \mathbb{Z}$. We call $c_{V}$ the central charge of $V$. We require that $L_{(0)}$ is diagonalizible and it defines the $\frac12 \mathbb{Z}$ grading - its eigenvalues are called (conformal) weights.  In several examples we will encounter $\frac12 \mathbb{Z}$-graded vertex superalgebras without a conformal vector.
For this reason, we define the character or graded dimension
as
$$ {\rm ch}[V](q)=\sum_{m \in \frac12 \mathbb{Z}} {\rm dim}(V_{(m)}) q^m.$$
As we do not care about modularity here, we suppress the $q^{-\frac{c}{24}}$ factor and also view $q$ as a formal variable.

\end{dfn}

\begin{ex}\cite{lepowsky2012introduction} We let $Vir$ denote the Virasoro Lie algebra.
Then the universal $Vir$-module $V_{Vir}(c,0)$ has a natural vertex operator algebra with central charge $c$.
\end{ex}

\begin{ex}\cite{kac1998vertex}
The universal vertex superalgebra associated with the $N=1$ Neveu-Schwarz Lie superalgebra will be denoted by
$V_{c}^{N=1}$, where $c$ is the central charge. It is a vertex operator superalgebra strongly generated by an odd vector $G_{(-\frac{3}{2})}\mathbf{1}$ and the conformal vector $L_{(-2)}\mathbf{1}$.
\end{ex}

\begin{ex}\cite{kac1998vertex},
The universal vertex superalgebra associated with the $N=2$ superconformal Lie algebra will be denoted by $V_{c}^{N=2}$.
It is a vertex operator superalgebra strongly generated by two odd vectors $G^{+}_{(-\frac{3}{2})}\mathbf{1}$, $G^{-}_{(-\frac{3}{2})}\mathbf{1}$ and two even vectors $L_{(-2)}\mathbf{1}$, $J_{(-1)}\mathbf{1}$.
\end{ex}

\begin{dfn}
A commutative vertex superalgebra $V$ is called a  vertex Poisson superalgebra if it is 
equipped with a linear operation,
$$V\rightarrow {\rm Hom}(V,z^{-1}V[z^{-1}]),\quad a\rightarrow Y_{-}(a,z)=\sum_{n\geq 0}a_{(n)}z^{-n+1},$$
such that
\begin{itemize}
    \item $(Ta)_{n}=-na_{(n-1)}$,
    \item $a_{(n)}b=\sum_{j\geq 0} (-1)^{n+j+1}\frac{(-1)^{p(a)p(b)}}{j!}T^{j}(b_{(n+j)}a),$
    \item $[a_{(m)},b_{(n)}]=\sum_{j\geq 0} \binom{m}{j}(a_{(j)}b)_{(m+n-j)},$
    \item  $a_{(n)}(b\cdot c)=(a_{(n)}b)\cdot c+(-1)^{p(a)p(b)}b\cdot (a_{(n)}c),$
\end{itemize}
for $a,b,c\in V$ and $n,m \geq 0$.
\end{dfn} A vertex Lie superalgebra structure on $V$ is given by $(V,Y_{-},T)$. So we can also say that a vertex Poisson superalgebra is a commutative vertex superalgebra equipped with a vertex Lie superalgebra structure. In fact, we can obtain a vertex Poisson superalgebra from any vertex superalgebra through standard filtration or Li's filtration. Following  \cite{li2005abelianizing}, we can define a decreasing sequence of subspaces $\left\{F_{n}(V)\right\}$ of the superalgebra $V$, where for $n\in \mathbb{Z}$, $F_{n}(V)$ is linearly spanned by the vectors $$u_{(-1-k_{1})}^{(1)}\ldots u_{(-1-k_{r})}^{(r)}\bf{1}$$
for $r\geq 1$, $u^{(1)},\ldots,u^{(r)}\in V,$ $k_{1},\ldots,k_{r}\geq 0$ with $k_{1}+\ldots+k_{r}\geq n.$ Then \[ V=F_{0}(V)\supset F_{1}(V)\supset\ldots\]such that  \begin{align*}
&u_{(n)}v\in F_{r+s-n-1}(V)\quad {\rm for} \quad u\in F_{r}(V),v\in F_{s}(V), r,s\in\mathbb{N},n\in \mathbb{Z},\\
&u_{(n)}v\in F_{r+s-n}(V)\quad {\rm for} \quad u\in F_{r}(V), v\in F_{r}(V),r,s,n\in \mathbb{N}. \end{align*}

\noindent The corresponding associated graded algebra 
$gr^{F}(V)=\coprod_{n\geq 0} F_{n}(V)/F_{n+1}(V)$ is a vertex Poisson superalgebra . Its vertex Lie superalgebra structure is given by:
$$ T(u+F_{r+1}(V))=Tu+F_{r+2}(V)$$
$$ Y_{-}(u+F_{r+1}(V),z)(v+F_{s+1}(V))=\sum_{n\geq 0}(u_{(n)}v+F_{r+s-n+1}(V))z^{-n-1}$$
for $u\in F_{r}(z),v\in F_{s}(z)$ with $r,s\in \mathbb{N}$. For the standard filtration $\left\{G_{n}(V)\right\}$, we also have the associated graded vertex superalgebra $gr^{G}(V)$. In \cite[Proposition 2.6.1]{arakawa2012remark}, T.Arakawa showed that \[gr^{F}(V)\cong gr^{G}(V)\] as vertex Poisson superalgebras. Thus we sometimes drop the upper index $F$ or $G$ for brevity.

According to \cite{li2005abelianizing}, we know that $$F_{n}(V)=\left\{u_{(-1-i)}v|u\in V,i\geq 1,v\in F_{n-i}(V)\right\}.$$ In particular, $F_{0}(V)/F_{1}(V)=V/C_{2}(V)=R_{V}\subset gr^{F}(V)$ which is a Poisson superalgebra according to \cite{zhu1996modular}. Its Poisson structure is given by $$\overline{u}\cdot \overline{v}=\overline{u_{(-1)}v},$$
$$\left\{\overline{u},\overline{v}\right\}=\overline{u_{(0)}v}$$
 for $u,v\in V$ where $\overline{u}=u+C_{2}(V)$. It was shown in  \cite[Corallary 4.3]{li2005abelianizing} that $gr^{F}(V)$ is generated by $R_{V}$ a differential algebra. \color{red}\color{black} We compute $C_{2}$-algebra for some simple examples first.

\begin{ex} Following notation in Example \ref{fermion}, let $\mathcal{F}$ be a free fermionic vertex superalgebra associated with an $n$-dimensional superspace $A$. Clearly, the $C_{2}$-algebra of $\mathcal{F}$ is $$R_{\mathcal{F}}=\mathbb{C}[\overline{\phi^{1}_{(-\frac{1}{2})}\mathbf{1}},\ldots, \overline{\phi^{n}_{(-\frac{1}{2})}\mathbf{1}}]$$ where $\overline{\phi^{i}_{(-\frac{1}{2})}\mathbf{1}}$ is even (resp. odd) if $\phi^{i}$ is even (resp. odd) in $A$.
\end{ex}

\begin{ex}\label{aff sl}
According to \cite{feigin2011zhu}, for a simple affine vertex algebras $L_{\widehat{\mathfrak{g}}}(k,0)$, $k \in \mathbb{N}$, where $\mathfrak{g}$ is a simple Lie algebra, we have:
 \[R_{L_{\widehat{\mathfrak{g}}}(k,0)}=\mathbb{C}[ u^{1}_{(-1)} \mathbf{1}, u^{2}_{(-1)} \mathbf{1}, \ldots, u^{n}_{(-1) }\mathbf{1}]/\langle  U(\mathfrak{g})\circ ((e_{\theta})_{(-1)})^{k+1}\mathbf{1}\rangle,\] where $\left\{u^{1}, u^{2}\ldots, u^{n}\right\} $ is a basis of $\mathfrak{g}$, $\theta $ is the highest root of $\mathfrak{g}$ and $\circ$ represents the adjoint action. In particular, when $\mathfrak{g}=sl(2)$,
 $$R_{L_{\widehat{{sl_{2}}}}(k,0)}\cong \mathbb{C}[e,f,h]/\langle  f^{i}\circ e^{k+1}|0\leq i\leq 2k+2\rangle$$ where $e,f,h$ correspond to $e_{(-1)}\mathbf{1}, f_{(-1)}\mathbf{1}, h_{(-1)}\mathbf{1}$.
\end{ex}

\begin{ex} 
For any simple Virasoro algebras $L_{Vir}(c_{(p,p')},0)$,  where $c_{(p,p')}=1-\frac{6(p-p')^{2}}{p p'}$ where $p>p' \geq 2$ and $p,p'$ are coprime, according to
\cite{beilinson1991introduction,van2018chiral} its $C_{2}$-algebra  is isomorphic to $\mathbb{C}[x]/\langle x^{\frac{(p-1)(p'-1)}{2}}\rangle,$ where $x$ corresponds to $\omega=L_{(-2)}\mathbf{1}$.
\end{ex}

\begin{ex} The $C_{2}$-algebra of $V_{c}^{N=1}$ is $R_{V_{c}^{N=1}}=\mathbb{C}[x,\theta]$ where $x$  and $\theta$ correspond to even vector $L_{(-2)}\mathbf{1}$ and odd vector $G_{(-\frac{3}{2})}\mathbf{1}$, respectively.
\end{ex}

\begin{ex}  The $C_{2}$-algebra of $V_{c}^{N=2}$ is $R_{V_{c}^{N=2}}=\mathbb{C}[x,y,\theta_{1},\theta_{2}]$ where $x,y,\theta_{1},\theta_{2}$ correspond to $L_{(-2)}\mathbf{1}$, $J_{(-1)}\mathbf{1}$, $G^{+}_{(-\frac{3}{2})} \vac$ and $G^{-}_{(-\frac{3}{2})} \vac,$ respectively. Here $\theta_{1},\theta_{2}$ are odd variables.
\end{ex}

\section{Affine jet superalgebra}

Inspired by the definition of jet algebra, we may give an analogous definition of a jet superalgebra in the affine case. Here we closely follow \cite{arakawa2012remark}.

 Let $\mathbb{C}[x^{1},x^{2}, \ldots, x^{n}, \theta^{1},\ldots, \theta^{m}]$ be a polynomial superalgebra where  $$x^{1}
 ,x^{2}, \ldots, x^{n}$$ are ordinary variables  and  $$\theta^{1},\ldots, \theta^{m}$$ are odd variables, i.e. $(\theta^{i})^{2}=0$ for $1\leq i \leq m.$  
 Let $f_{1},f_{2}, \ldots, f_{n}$ be $\mathbb{Z}_{2}$-homogeneous elements in the polynomial superalgbera. We will define the jet superalgbra of the quotient superalgebra: \begin{align*}R=\frac{ \mathbb{C}[x^{1},x^{2}, \ldots, x^{n}, \theta^{1},\ldots, \theta^{m}]}{\langle  f_{1},f_{2}, \ldots, f_{r}\rangle}.\end{align*}

 Firstly, let us introduce new even variables $x^{j}_{(-\Delta_{j}-i)}$  and odd variables $\theta^{j'}_{(-\Delta_{j'}-i)}$ for $i=0,\ldots,m$ where $\Delta_{j}$ and $\Delta_{j'}$ are degrees of $x^{j}$ and $\theta^{j'}$. In most cases, we will assume that the degree of each variable is $1$ although in some cases the odd degree can be shifted by $\frac12$. We define an even derivation $T$ on \[\mathbb{C}[x^{j}_{(-\Delta_{j}-i)},\theta^{j'}_{(-\Delta_{j'}-i))}\;|\;0\leq i\leq m,\; 1\leq j\leq n,\; 1\leq j'\leq m]\] as \begin{equation*}
  T(x^{j}_{(-\Delta_{j}-i)}) =
    \begin{cases}
      (-\Delta_{j}-i)x^{j}_{(-\Delta_{j}-i-1))} & \text{for $i\leq m-1$}\\
      0 & \text{for $i=m$},\\
    \end{cases}
\end{equation*}
 and
 \begin{equation*}
  T(\theta^{j'}_{(-\Delta_{j'}-i)}) =
    \begin{cases}
      (-\Delta_{j'}-i)\theta^{j'}_{(-\Delta_{j'}-i-1))} & \text{for $i\leq m-1$}\\
      0 & \text{for $i=m$}.\\
    \end{cases}
\end{equation*} Here we identify $x^{j}$ and $\theta^{j'}$ with $x^{j}_{(-\Delta_{j})}$ and $\theta^{j'}_{(-\Delta_{j'})},$ respectively. Set \begin{align*}R_{m}=\frac{\mathbb{C}[x^{j}_{(-\Delta_{j}-i)},\theta^{j}_{(-\Delta_{j'}-i))}\;|\;0\leq i\leq m,\; 1\leq j\leq n,\; 1\leq j'\leq m]}{\langle T^{j}f_{i}|i=1,\ldots n,\; j\in\mathbb{N}\rangle } .  \end{align*} Then the $m$-jet superscheme $V_{m}$ is defined as ${\rm Spec}((R_{m})_{\overline{0}})$ where $(R_{m})_{\overline{0}}$ is the even part of the $R_{m}.$ The infinite jet superalgebra of $V$ is \begin{align*} J_{\infty}(R)&=\displaystyle\lim_{\underset{m}{\leftarrow}}R_{m}\\&=\frac{\mathbb{C}[x^{j}_{(-\Delta_{j}-i)},\theta^{j'}_{(-\Delta_{j'}-i))}\;|\;0\leq i,\; 1\leq j\leq n,\; 1\leq j'\leq m]}{\langle T^{j}f_{i}|i=1,\ldots n,\; j\in\mathbb{N}\rangle }.\end{align*} We often omit "infinite" and call it jet superalgebra for brevity. The jet superalgebra is a differential  commutative superalgebra. We denote the ideal
\[\langle T^{j}f_{i};i=1,\ldots,n, j\geq 0\rangle\] by $\langle f_{1},\ldots,f_n \rangle_{\partial}.$  Later, we sometimes write $x_{(j)}$ as $x(j).$ The infinite jet superscheme, or arc space, is defined as  $$J_{\infty}(V)={\rm Spec}((J_{\infty}(R))_{\overline{0}}).$$ We define the degree of each variable $u_{(-\Delta-j)}$ to be $\Delta+j$ where $u=x$ or $\theta$. Then $J_{\infty}(R)=\displaystyle \coprod_{\frac{1}{2}\mathbb{Z}}(J_{\infty}(R))_{(m)}$ where $(J_{\infty}(R))_{(m)}$ is the set of all elements in jet superalgebra with degree $m$. We define Hilbert series of  $J_{\infty}(R)$ as: \[HS_{q}(J_{\infty}(R))=\sum_{m \in \frac12 \mathbb{Z}} {\rm dim}(J_{\infty}(R)_{(m)}) q^m.\]

Following \cite{arakawa2012remark}, $J_{\infty}(R)$ has a unique Poisson vertex superalgebra structure such that

$$u_{(n)}v=\begin{cases} \left\{u,v\right\}, & \mbox{if } n=0 \\ 0, & \mbox{if } n>0 \end{cases} $$
for $u,v\in R\in J_{\infty}(R)$.

Furthermore, we can extend the embedding $R_{V}\hookrightarrow gr^{F}(V)$ to a surjective differential superalgebra homomorphism $J_{\infty}(R_{V}) \twoheadrightarrow gr^{F}(V)$. It is obvious that the map is a differential superalgebra homomorphism. It is surjective since  $gr^{F}(V)$ is generated by $R_{V}$ as a differential algebra. Moreover, it was shown in \cite{arakawa2012remark} that this map is actually a Poisson vertex superalgebra epimorphism. From now on, we call this map $\psi.$ The map $\psi$ is not necessarily injective and it is an open problem to characterize rational vertex algebras for which $\psi$ is injective.

\subsection{Complete lexicographic ordering}\label{ordering} 
Following \cite{feigin2001combinatorics}, we define the complete lexicographic ordering on a basis or spanning set of the jet superalgebra.  Given a jet superalgebra \begin{align*} J_{\infty}(\mathbb{C}[y^1,y^2,\ldots,y^n]/I)=\mathbb{C}[y^{1}_{(-\Delta_1-i)},\dots,y^{n}_{(-\Delta_n-i)}|i\in\mathbb{N}]/I_{\infty} \end{align*} where $\Delta_{i}$ is the degree of $y^{i},$ we can first define an ordering of all variables in the following way: \[y^{1}_{(-\Delta_1)}<y^{2}_{(-\Delta_1)}<\ldots<y^{n}_{(-\Delta_n)}<y^{1}_{(-\Delta_1-1)}<y^{2}_{(-\Delta_2-1)}<\ldots.\]
\begin{dfn}
A monomial $u$ of $ J_{\infty}(\mathbb{C}[y^1,y^2,\ldots,y^n]/I) $ is called an ordered monomial if it is of the form:\[(y^{n}_{(-\Delta_n-m)})^{a^{n}_{m+1}}\ldots(y^{1}_{(-\Delta_1-m)})^{a^1_{m+1}}\ldots (y^{n}_{(-\Delta_n)})^{a^{n}_{1}}\ldots (y^{2}_{(-\Delta_1)})^{a^{2}_{1}}(y^{1}_{(-\Delta_1)})^{a^1_1}\] where $m\in \mathbb{Z}_{+}$ and $a^{i}_j\in \mathbb{N}$.
\end{dfn}
It should be clear that all ordered monomials form a spanning set of the jet superalgebra. Then let us define the multiplicity of an ordered monomial as $$\mu(u)=\sum_{i=1}^{m+1} (a^{1}_i+a^{2}_i+\ldots+a^{n}_i).$$ Given two arbitrary ordered monomials \[u=y^{n}_{(-\Delta_n-m)})^{a^{n}_{m+1}}\ldots(y^{1}_{(-\Delta_1-m)})^{a^1_{m+1}}\ldots (y^{n}_{(-\Delta_n)})^{a^{n}_{1}}\ldots (y^{2}_{(-\Delta_1)})^{a^{2}_{1}}(y^{1}_{(-\Delta_1)})^{a^1_1}\] and \[v=y^{n}_{(-\Delta_n-m)})^{b^{n}_{m+1}}\ldots(y^{1}_{(-\Delta_1-m)})^{b^1_{m+1}}\ldots (y^{n}_{(-\Delta_n)})^{b^{n}_{1}}\ldots (y^{2}_{(-\Delta_1)})^{b^{2}_{1}}(y^{1}_{(-\Delta_1)})^{b^1_1},\] we define a complete lexicographic ordering as following: If $\mu(u)<\mu(v)$, we say that $u<v$. If $\mu(u)=\mu(v)$, we compare exponents of \[y^{1}_{(-\Delta_1)},y^{2}_{(-\Delta_1)},\ldots,y^{n}_{(-\Delta_n)},\ldots,y^{1}_{(-\Delta_1-m)},y^{n}_{(-\Delta_n-m)}\] in this order. Namely, we say $v<u$ if $a^1_1<b^1_1$; if they are equal, we then compare $a^{2}_{1}$ and $b^2_1$, and so on.  Given a polynomial $f$, we call the greatest monomial among all its terms with respect to the complete lexicographic ordering the {\em  leading term}  of $f$.

\section{Affine and lattice vertex algebras}

In this section we analyze the Poisson (super)algebra $R_V$ and the injectivity of the $\psi$ map for some familiar examples of affine and lattice vertex algebras.

\begin{ex}
It was shown in \cite[Proposition 2.7.1]{arakawa2012remark} that for any simple Lie algebra $\mathfrak{g}$, we have $J_{\infty}(R_{V_{\widehat{\mathfrak{g}}(k,0)}})\cong gr^{F}(V_{\widehat{\mathfrak{g}}(k,0)})$.
\end{ex}

\begin{prop}\label{iso fermion} For the free fermionic vertex superalgebra,
$J_{\infty}(R_{\mathcal{F}})\cong gr^{F}(\mathcal{F})$ as vertex Poisson superalgebras.
\end{prop}
\begin{proof}
We use T.Arakawa's argument in \cite[Proposition 2.7.1]{arakawa2012remark}. We include the proof for completeness. Here we still follow notation from Example \ref{fermion}.
According to \cite[Section 3.6]{kac1998vertex}, we can choose a conformal vector such that $\mathcal{F}$ is $\frac{1}{2}\mathbb{Z}_{\geq 0}$-graded. We consider the standard filtration $G$ on $F$. Firstly, we have $\mathcal{F}\cong U(A[t^{-1}]t^{-1})$ as  super vector spaces. Moreover $$G^{n}(\mathcal{F})=\left\{u^{1}_{(-k_{1})}\ldots u^{r}_{(-k_{r})}\mathbf{1}\;|\;k_{i}\in \frac{1}{2}\mathbb{Z}_{\geq 0},\;r\geq 0, \frac{r}{2}\leq n\right\}$$  where $u^{i}
\in \left\{\phi^{1},\ldots, \phi^{n}\right\}$. So $gr^{G}(\mathcal{F})\cong S(A[t^{-1}]t^{-1})\cong J_{\infty}(R_{\mathcal{F}}) $ as Poissson vertex superalgebra. Therefore $gr^{G}(\mathcal{F})\cong gr^{F}(\mathcal{F})\cong J_{\infty}(R_{\mathcal{F}}) $.
\end{proof}
Similarly, we can show that $\psi$ is an isomorphism for vertex superalgebra $V_{\widehat{{\mathfrak{g}}}}(k,0)$ where $\mathfrak{g}$ is a Lie superalgebra satisfying conditions in Example \ref{aflie} , and for superconformal vertex algebras $V_{c}^{N=1}$ and $V_{c}^{N=2}$.
\vspace{0.2cm}

Let  
$$V_{ \sqrt{p}\mathbb{Z}}=M(1) \otimes \mathbb{C}[\sqrt{p}\mathbb{Z}],$$
be a rank one lattice vertex algebra (resp. superalgebra) constructed from an integral
lattice $L=\mathbb{Z}\alpha \cong \sqrt{p}\mathbb{Z}$ where $\langle \alpha,\alpha\rangle=p$ is even (resp. odd). It has a conformal vector 
$\omega=\frac{1}{2p}\al_{(-1)}^2 \mathbf{1}$.
As usual, we denote the extremal lattice vectors by $e^{n \alpha}$, $n \in \mathbb{Z}$.

\begin{prop}\label{lattice}
For the lattice vertex algebra $V_{\sqrt{p}\mathbb{Z}}$  we have $$R_{V_{\sqrt{p}\mathbb{Z}}}\cong \mathbb{C}[x,y,z]/\langle  x^{2},y^{2},xy=z^{p},xz,yz\rangle. $$ When $p$ is odd, $x$ and $y$ are odd vectors.
\end{prop}

\begin{proof}
According to the following calculations, \begin{align*} &(e^{\alpha})_{(-2)}(e^{-\alpha})-\frac{(b_{(-1)})^{p+1}\textbf{1}}{(p+1)!}
 \in R_{V_{\sqrt{p}\mathbb{Z}}}\\  &(e^{\alpha})_{(-2)}(\textbf{1})=b_{(-1)}e^{\alpha}\\
&(e^{\alpha})_{(-p-1)}(e^{\alpha})=e^{2\alpha} \\
&(e^{-\alpha})_{(-2)}(\textbf{1})=-b_{(-1)}e^{-\alpha}\\
&(e^{-\alpha})_{(-p-1)}(e^{-\alpha})=2e^{-2\alpha},\end{align*}

\noindent we know that all vectors except for $b_{(-1)}\textbf{1}$, \ldots, $b_{(-1)}^{p}\textbf{1}$, $e^{\alpha}$, $e^{-\alpha}$ and $\textbf{1}$ are zero in $R_{V_{\sqrt{p}\mathbb{Z}}}$.

Then we will show that all those vectors are indeed nonzero in $R_{V_{\sqrt{p}\mathbb{Z}}}$.
Suppose there exist $a,b\in V_{\sqrt{p}\mathbb{Z}}$ such that \[a_{(-n)}b- e^{\alpha}\in R_{V_{\sqrt{p}\mathbb{Z}}}\] where $-n\geq 2$. Then $ wt(a_{n}b)=wt(a)+wt(b)-n-1=\frac{p}{2}$ which implies that 
$a,b\in \pi_{{0}}$ where $\pi_{0}$ is the  Heisenberg subalgebra $\mathbb{C}[\al_{(n)}]_{n<0}\cdot \mathbf{1}.$
This is a contradiction. So the equivalent class $\overline{e^{\alpha}}$  is nonzero in $R_{V_{\sqrt{p}\mathbb{Z}}}$. By using similar weight argument, we can show that equivalence classes $$\overline{e^{-\alpha}}, \overline{\textbf{1}}, \overline{b_{(-1)}\textbf{1}}, \ldots, \overline{b_{(-1)}^{p}\textbf{1}}$$ are all nonzero in $R_{V_{\sqrt{p}\mathbb{Z}}}$.  Moreover, we have \[(e^{\alpha})_{(-1)}(e^{-\alpha})-\frac{b_{(-1)}^{p}\textbf{1}}{p!}\in R_{V_{\sqrt{p}\mathbb{Z}}}.\] Then  the map $\psi$ is sending $\overline{e^{\alpha}}$ to $x$, $\overline{e^{-\alpha}}$ to $y$,  $\overline{\textbf{1}}$ to 1 and $\sqrt[p]{\frac{1}{p!}}\overline{b_{(-1)}\textbf{1}}$ to $z$ induces an isomorphism of algebras.

\end{proof}
\begin{rem} According to the Frenkel-Kac construction, we know that $V_{\sqrt{2}\mathbb{Z}}\cong L_{\widehat{sl_{2}}}(1,0)$.  Following Proposition \ref{lattice}, we have $R_{L_{\widehat{sl_{2}}}(1,0)}\cong \mathbb{C}[e,f,h]/\langle  e^{2},f^{2},ef=h^{2},eh,fh\rangle$. \end{rem}
Before we move on, let us briefly recall definition of the associative Zhu algebra. Given a vertex superalgebra $V=\displaystyle\coprod_{n\in \frac{1}{2}\mathbb{Z}}V_{n}$ where $V_{\overline{0}}=\coprod_{n\in \mathbb{Z}}V_{n}$ and $V_{\overline{1}}=\coprod_{n\in \mathbb{Z}+\frac{1}{2}}V_{n},$ there are two binary operations defined as following:
for homegeneous $a\in V$ and $b\in V$,
$$a\ast b=\begin{cases} \displaystyle \sum_{i\geq 0}\binom{wt(a)}{ i}a_{(i-1)}b, & \mbox{if } a,b\in V_{\overline{0}} \\ 0, & \mbox{if } a\;\text{or}\; b\in V_{\overline{1}} \end{cases}$$ and $$a\circ b=\begin{cases} \displaystyle \sum_{i\geq 0}\binom{wt(a)}{i}a_{(i-2)}b, & \mbox{if } a\in V_{\overline{0}} \\ \displaystyle \sum_{i\geq 0}\binom{wt(a)-\frac{1}{2}}{i}a_{(i-1)}b, & \mbox{if } a\in V_{\overline{1}} \end{cases}.$$ Let $O(V)$ be the linear span of elements of the form $a\circ b$ in $V.$ Then Zhu's algebra $A(V)$ is defined as the quotient space $V/O(V)$ with the mutiplication from $\ast$. According to \cite{abe2007mathbb}, there is a filtration $\left\{\overline{F}_{k}(A(V))\right\}$ on $A(V)$ where $\overline{F}_{k}(A(V)):=(\displaystyle\bigoplus_{i=0}^{k}V_{i}+O(V))/O(V)$. Its associated graded algebra $$gr^{\overline{F}}(A(V))=\displaystyle\bigoplus_{i=0}^{\infty}\overline{F}_{k}(A(V))/\overline{F}_{k-1}(A(V))$$ is a commutative algebra. Now we can prove:
\begin{cor}
Let $p$ be a positive odd integer, then the even part of $R_{V_{\sqrt{p}\mathbb{Z}}}$, i.e. $(R_{V_{\sqrt{p}\mathbb{Z}}})_{\overline{0}}$, is isomorphic to the associated graded algebra $gr^{\overline{F}}(A_{V_{\sqrt{p}\mathbb{Z}}})$.

\end{cor}

\begin{proof}
 According to  \cite[Theorem 3.3]{ogawa2000zhu}, we know that $A_{V_{\sqrt{p}\mathbb{Z}}}\cong \mathbb{C}[x]/(F_{p}(x))$ where $F_{p}(x)=x(x+1)(x-1)\ldots(x+\frac{(p-1)}{2})(x-\frac{(p-1)}{2})$ in which $x$ corresponds to $[\al_{(-1)}\mathbf{1}]$ in $A_{V_{\sqrt{p}\mathbb{Z}}}$. According to \cite{abe2007mathbb}, we have an epimorphism \[f:R_{V_{\sqrt{p}\mathbb{Z}}}\twoheadrightarrow gr^{\overline{F}}(A_{V_{\sqrt{p}\mathbb{Z}}})\] given by $f(\overline{a}\cdot\overline{b})=[a\ast b]+\overline{F}^{k+l-1}(A(V))$ for $a\in V_k$ and $b\in V_l$.
  Then according to Proposition \ref{lattice}, we have $$(R_{V_{\sqrt{p}\mathbb{Z}}})_{\overline{0}}\cong gr^{\overline{F}}(A_{V_{\sqrt{p}\mathbb{Z}}})\cong \mathbb{C}[x]/\langle  x^{p}\rangle  $$ via $f$.

\end{proof}
\begin{rem}
If $L=\sqrt{2k}\mathbb{Z}$ $(k\geq 1)$ is an even lattice, above result is not true. According to \cite{dong1997certain}, we have $gr^{\overline{F}}(V_{\sqrt{2k}\mathbb{Z}})\cong \mathbb{C}[x]/\langle x^{2k-1} \rangle$ which is obviously not isomorphic to $R_{V_{\sqrt{2k}\mathbb{Z}}}.$
\end{rem}

In \cite{van2018chiral}, authors proved the map $\psi$ is an isomorphism for $L_{\widehat{sl_{2}}}(k,0)$  by using a PBW-type  basis of $L_{\widehat{sl_{2}}}(k,0)$ from \cite{meurman1987annihilating} and Gr\"{o}bner bases. In \cite{feigin2009pbw}, author essentially proved the same result by using a technique called "degeneration procedure". In the following, we briefly explain how his results proves isomorphism.

\begin{prop}\label{A1}
The map $\psi:J_{\infty}(R_{L_{\widehat{sl_{2}}}(k,0)})\cong gr^{F}(L_{\widehat{sl_{2}}}(k,0))$ is an isomorphism of vertex Poisson algebras.
\end{prop}
\begin{proof}
Let us prove $k=1$ case. It is clear that $\psi(u_{(-i)})=\overline{u_{(-i)}\mathbf{1}}$ for $u\in \left\{e,f,h\right\}$ and $i>0$. Let $u(z)=\sum_{n\leq -1}u_{(n)}z^{-n-1}$ where $u\in \left\{e,f,h \right\}$. Now we consider $e(z)e(z)$. The coefficient of $z^{n}$ equals $T^{n}(e_{(-1)}e_{(-1)})$ up to a scalar multiple for $n\geq 0.$ And we have similar results for $f^{2}$, $ef=h^{2}$, $eh$ and $fh.$ Thus
\begin{align*} J_{\infty}(R_{L_{\widehat{sl_{2}}}(1,0)})\cong \frac{\mathbb{C}[e_{(-1-i)},f_{(-1-i)},h_{(-1-i)}\;|\;i\in\mathbb{N}]}{\langle  e(z)^{2},f(z)^{2},e(z)f(z)=h(z)^{2},e(z)h(z),f(z)h(z)\rangle} \end{align*}  where \[\langle  e(z)^{2},f(z)^{2},e(z)f(z)=h(z)^{2},e(z)h(z),f(z)h(z)\rangle\] means the ideal generated by the Fourier coefficients of $e(z)^{2},f(z)^{2},e(z)f(z)=h(z)^{2},e(z)h(z),f(z)h(z)$. Our result now follows from the above argument and \cite[Corollary 2.3]{feigin2009pbw}. When $k\geq2$, the result follows from the same argument and \cite[Theorem 3.1]{feigin2009pbw}.

\end{proof}

Before we prove next result, let us fix some notation first. We denote a simple finite-dimensional Lie algebra of type $C_{n},n\geq 2$ by $\mathfrak{g}$. Here we assume that $\mathfrak{g}$ has a basis 
$\left\{b^{i}|1\leq i\leq (2n+1)n\right\}. $ Let $\theta$ be the maximal root of $\mathfrak{g}$, and $x_{\theta}$ the corresponding maximal root vector. We let $\widehat{\mathfrak{g}}$ be the affine Lie algebra associated with $\mathfrak{g}$ and its universal vacuum representation is $V_{\widehat{\mathfrak{g}}}(k,0)$ for $k\in \mathbb{Z}_{>0}.$ Set \[R=U(\mathfrak{g})\circ (x_{\theta})_{(-1)}^{k+1}\textbf{1},\quad \overline{R}=\mathbb{C}- {\rm Span} \left\{r_{(n)}|r\in R,\;n\in \mathbb{Z}\right\}\] where $\circ$  is the adjoint action. Then $\widehat{\mathfrak{g}}$-module $V_{\widehat{\mathfrak{g}}}(k,0)$ has a maximal submodule $I_{\widehat{\mathfrak{g}}}(k,0)$ generated by $\overline{R}\cdot \bf{1}.$ And $$L_{\widehat{\mathfrak{g}}}(k,0)=V_{\widehat{\mathfrak{g}}}(k,0)/I_{\widehat{\mathfrak{g}}}(k,0).$$ Now we are ready to prove:

 \begin{thm}
 The map $\psi$ is an isomorphism for the affine vertex algebra  $L_{\widehat{\mathfrak{g}}}(1,0).$

 \end{thm}

 \begin{proof}

  It is clear that the $C_{2}$-algebra of $L_{\widehat{\mathfrak{g}}}(1,0)$ is $R_{L_{\widehat{\mathfrak{g}}}(1,0)}=S(\mathfrak{g})/\langle  U(\mathfrak{g})\circ e_{\theta}^{2}\rangle$ where $S(\mathfrak{g})$ is the symmetric algebra of $\mathfrak{g}$ and $U(\mathfrak{g})$ is the universal enveloping algebra of $\mathfrak{g}$. We denote the algebra $$\mathbb{C}[b^{i}_{(-j)}|j>0]/\langle  U(\mathfrak{g})\circ e_{\theta}^{2}(z)\rangle$$  by
  $Q$ where $e_{\theta}(z)=\sum_{n<0} (e_{\theta})_{(n)}z^{-n-1}$. Following the similar argument in Proposition \ref{A1}, we see that $J_{\infty}(R_{L(\Lambda_{0})})\cong Q$. In order to show that $\psi$ is an isomorphism, it is enough to prove that $gr^{F}(L(\Lambda_{0}))$ and $Q$ have the same basis. Notice that \begin{align*}I=\overline{R} \cap \mathbb{C}[b^{i}_{(-j)}|j>0]=\langle  U(\mathfrak{g})\circ e_{\theta}^{2}(z)\rangle.\end{align*} We can define an order on all monomials of $\mathbb{C}[b^{i}_{(-j)}|j>0]$ in the sense of  
  \cite[Section 8]{primc2016combinatorial}. From the same paper, we know that every nonzero homogeneous polynomial $\mathbb{C}[b^{i}_{(-j)}|j>0]$ has a unique largest monomial. For an arbitrary nonzero  polynomial $u$, we define the leading term $lt(u)$ as the largest monomial of the nonzero homogeneous component of the smallest degree, which is unique. We denote all monomials in $\mathbb{C}[b^{i}_{(-j)}|j>0]$ by $\mathcal{P}.$ We clearly have $\mathcal{P}$ as a spanning set of $Q$. Since $u=0$ in $Q$ if $u\in I$, the leading term $lt(u)$ \color{red}\color{black} equals the linear combination of other terms. Therefore $\mathcal{P}\setminus \langle lt(U)\rangle$ is a smaller spanning set of $Q$. And we denote it by $\mathcal{RR}$. Meanwhile according to \cite[Theorem 11.3]{primc2016combinatorial}, we know that $\psi(\mathcal{RR})$ is a basis of $gr(L(\Lambda_{0}))$. Together with the surjectivity of $\psi$, we have that $\mathcal{RR}$ is a basis of $L_{\widehat{\mathfrak{g}}}(1,0)$. Therefore $\psi$ is an isomorphism.

 \end{proof}

\subsection{\texorpdfstring{$N=2$}{Lg} vertex superalgebra at \texorpdfstring{$c=1$}{Lg}}\label{N2}

In this section we study the simple $N=2$ superconformal vertex algebra of central charge $c=1$, denoted by $L_{1}^{N=2}$.
The odd lattice vertex algebra $V_{\sqrt{3}\mathbb{Z}}$ is known to be isomorphic to $L_{1}^{N=2}$. Here we identify $\frac{1}{3}b_{(-1)}\mathbf{1} $ with $J_{(-1)}\mathbf{1}$,
$\frac{1}{\sqrt{3}}e^{\pm\alpha} $ with $G_{(\pm \frac{3}{2})}\mathbf{1}$ and $\frac{1}{6}(b_{(-1)}b_{(-1)}\bf{1} )_{(-1)}\mathbf{1} $ with $L_{(-2)}\mathbf{1}$.

 According to \cite{adamovic1999rationality,adamovic2001vertex}, the maximal submodule of $V_1^{N=2}$ is generated by \[G^{+}_{(-\frac{5}{2})}G^{+}_{(-\frac{3}{2})}\mathbf{1}\quad {\rm and} \quad G^{-}_{(-\frac{5}{2})}G^{-}_{(-\frac{3}{2})}\mathbf{1}.\]
 By identifying $G^{+}$ with $G^{+}_{(-\frac{3}{2})}\mathbf{1}$, $G^{-}$ with $G^{-}_{(-\frac{3}{2})}\mathbf{1}$ and $h$ with
$J_{(-1)}\mathbf{1},$ we have $$R_{L^{N=2}_1}\cong \mathbb{C}[G^{+},G^{-},h]/\langle  (G^{+})^{2},(G^{-})^{2},G^{+}G^{-}=h^{3},G^{+}h,G^{-}h\rangle. $$
For the jet superalgebra of \[\mathbb{C}[G^{+},G^{-},h]/\langle  (G^{+})^{2},(G^{-})^{2},G^{+}G^{-}=h^{3},G^{+}h,G^{-}h\rangle,\] we identify $G^{+},G^{-},h$ with $G^{+}(-\frac{3}{2}),G^{-}(-\frac{3}{2}),h(-1)$. And \begin{align*}J_{\infty}(R_{V_{\sqrt{3}\mathbb{Z}}})\cong\frac{\mathbb{C}[G^{+}(-\frac{3}{2}-i),G^{-}(-\frac{3}{2}-i),h(-1-i)|i\in\mathbb{N}]}{\langle  (G^{+}(z))^{2},(G^{-}(z))^{2},G^{+}(z)G^{-}(z)=h(z)^{3},G^{+}(z)h(z),G^{-}(z)h(z)\rangle}\end{align*}where $G^{\pm}(z)=\sum_{n\leq-\frac{3}{2}} G^{\pm}(n)z^{-n-\frac{3}{2}}$, $h(z)=\sum_{n\leq -1}h(n)z^{-n-1}.$ The map $\psi$ is not an isomorphism in this case because the images of nonzero elements \[G^{+}(-\frac{5}{2})G^{+}(-\frac{3}{2}) \quad {\rm and} \quad G^{-}(-\frac{5}{2})G^{-}(-\frac{3}{2})\] in the jet superalgebra under $\psi$, i.e. $G^{+}_{(-\frac{5}{2})}G^{+}_{(-\frac{3}{2})}\mathbf{1}$ and $G^{-}_{(-\frac{5}{2})}G^{-}_{(-\frac{3}{2})}\mathbf{1}$, are null vectors. Thus \[\langle a,b\rangle_{\partial}= \langle T^{i}(G^{+}(-\frac{5}{2})G^{+}(-\frac{3}{2})),T^{i}(G^{-}(-\frac{5}{2})G^{-}(-\frac{3}{2}))|i\geq 0\rangle\subset ker(\psi)\] where $a=G^{+}(-\frac{5}{2})G^{+}(-\frac{3}{2})$ and $b=G^{-}(-\frac{5}{2})G^{-}(-\frac{3}{2}).$

Let us consider \[J_{\infty}(R_{L^{N=2}_1})/\langle a,b\rangle_{\partial} .\] We will write down a spanning set of $J_{\infty}(R_{L^{N=2}_1})/\langle a,b\rangle_{\partial}.$ We let the ordered monomial be a monomial of the form
\begin{align*} &{G^{-}(-n-\frac{1}{2})}^{a_{n}} h(-n)^{b_{n}}{G^{+}(-n-\frac{1}{2})}^{c_{n}}\ldots {G^{-}(-\frac{5}{2})}^{a_{2}}h(-2)^{b_{2}} {G^{+}(-\frac{5}{2})}^{c_{2}}\\ &{G^{-}(-\frac{3}{2})}^{a_{1}}h(-1)^{b_{1}}{G^{+}(-\frac{3}{2})}^{c_{1}}.\end{align*} Then we have a complete lexicographic ordering on the set of ordered monomials in the sense of Section \ref{ordering}.  
  Now let us find the leading terms of the Fourier coefficients of \[G^{+}(z)G^{-}(z)=h^{3}(z),\quad G^{+}(z)h(z),\quad G^{-}(z)h(z),\quad T^{i}(a),\quad T^{i}(b).\]

\begin{itemize}
    \item[(a)] Leading term of $G^{+}(z)h(z)$: \begin{itemize}
        \item  $n$ is even, the leading term of the coefficient of $z^{n}$
 is \[h(\frac{-2-n}{2})G^{+}(-\frac{3+n}{2}).\]
 \item  $n$ is odd, the leading term of the coefficient of $z^{n}$ is $$G^{+}(\frac{-4-n}{2})h(\frac{-1-n}{2}).$$\end{itemize}
   \item[(b)] Leading term  of $G^{-}(z)h(z)$: \begin{itemize} \item  $n$ is even, the leading term of the $z^{n}$-th coefficient is \[G^{-}(-\frac{3+n}{2})h(\frac{-2-n}{2}).\]\item   $n$ is odd, the leading term of the $z^{n}$-th coefficient is \[h(\frac{-3-n}{2})G^{-}(\frac{-2-n}{2}).\] \end{itemize}
   \item[(c)]  Leading term  of $G^{+}(z)G^{-}(z)=h^{3}(z)$: \begin{itemize} \item $n=1$, the leading term of the coefficient of $z$ is \[h(-1)h(-1)h(-1).\] \item $n$ is even, the leading term of the $z^{n}$-th coefficient is \[G^{-}(-\frac{3+n}{2})G^{+}(-\frac{3+n}{2}).\]  \item $n$ is odd, the leading term of the $z^{n}$-th coefficient is \[G^{-}(\frac{-n-4}{2})G^{+}(\frac{-n-2}{2}).\]\end{itemize}
   \item[(d)]  Leading term of $T^{n}(a)$ or $T^{n}(b)$: \begin{itemize}\item  $n$ is even, the leading term of the coefficient of $z^{n}$ is \[G^{\pm}(\frac{-n-5}{2})G^{\pm}(\frac{-n-1}{2}).\] \item $n$ is odd, the leading term of the coefficient of $z^{n}$ is  \[G^{\pm}(\frac{-n-4}{2})G^{\pm}(\frac{-n-2}{2}).\] \end{itemize}
\end{itemize}

Clearly all ordered monomials constitute a spanning set of $J_{\infty}(R_{L_{1}^{N=2}})/\langle a,b\rangle_{\partial}$. Since all polynomials we considered above equal zero in $J_{\infty}(R_{L_{1}^{N=2}})/\langle a,b\rangle_{\partial}$,
the leading term of each can be written as a linear combination of all other terms. Thus if we want to get a "smaller" spanning set, all above leading terms can not \color{red}\color{black} appear as segments of an ordered monomial. Therefore we can impose some difference conditions on ordered monomials by using these leading terms to get a new spanning set.

\begin{dfn}
 We call an ordered monomial a $Gh-$ monomial, if it satisfy the following conditions:
 \begin{itemize}
     \item[$(\romannumeral 1)$] Either $b_i$ or $c_i$ is $0$ and either $b_i$ or $c_{i+1}$ is $0.$
     \item[$(\romannumeral 2)$]  Either $a_{i}$ or $b_{i}$ is $0$ and either $a_i$ or $b_{i+1}$ is $0.$
     \item[$(\romannumeral 3)$]  $a_{i}+c_{i}+c_{i+1}\leq 1$ and $b_{1}\leq 2$, $i\geq 2.$
     \item[$(\romannumeral 4)$]  $c_{i}+c_{i+2}+c_{i+1}\leq 1$,\; $a_{i}+a_{i+2}+a_{i+1}\leq 1$.
 \end{itemize}
 \end{dfn}

Here constraints $(\romannumeral 1)$-$(\romannumeral 4)$ are coming from leading terms in $(a)$-$(d)$, respectively.
Then we have the following:

 \begin{prop}
 $Gh-$monomials form a spanning set of $$A=J_{\infty}(R_{L_{1}^{N=2}})/\langle a,b\rangle_{\partial}.$$
 \end{prop}
Let us write down the first few terms of the Hilbert series of $A$.

 \begin{ex}
 For $i\leq 5$, $Gh-$monomials give us basis of $A_{i}$:
 \begin{align*}
& A_{1}:h(-1) \\
& A_{\frac{3}{2}}:G^{+}(-\frac{3}{2}), G^{-}(-\frac{3}{2}) \\
& A_{2}: h({-1})^{2}, h({-2})\\
& A_{\frac{5}{2}}: G^{+}(-\frac{5}{2}), G^{-}(-\frac{5}{2}) \\
& A_{3}: G^{-}(-\frac{3}{2})G^{+}(-\frac{3}{2}), h({-1})h({-2}),h({-3}) \\
& A_{\frac{7}{2}}: G^{-}({-\frac{5}{2}})h({-1}), G^{+}({-\frac{7}{2}}), G^{-}({-\frac{7}{2}}), G^{+}({-\frac{3}{2}})h({-2}) \\
&  A_{4}: G^{+}({-\frac{3}{2}})G^{-}({-\frac{5}{2}}),h({-2})^{2}, h({-1})^{2}h({-2}),h({-3})h({-1}),h({-4}) \\
& A_{\frac{9}{2}}: G^{-}({-\frac{5}{2}})h({-1})^{2}, G^{+}({-\frac{9}{2}}), G^{-}({-\frac{9}{2}}), \\
&  \ \ \ \ \ \ \ G^{-}({-\frac{7}{2}})h({-1}), G^{+}({-\frac{7}{2}})h({-1}), h({-3})G^{-}({-\frac{3}{2}}), h({-3})G^{+}({-\frac{3}{2}})  \\
& A_{5}:  G^{-}({-\frac{3}{2}})G^{+}({-\frac{7}{2}}), G^{+}({-\frac{3}{2}})G^{-}({-\frac{7}{2}}), h({-1})h({-4}),\\
& \ \ \ \ \ \ \ \ h({-1})h({-2})^{2}, h({-1})^{2}h({-3}),h({-2})h({-3}),h({-5} ).
\end{align*}

 \end{ex}
\noindent We have $HS_{q}({A})=1+q+2q^{\frac{3}{2}}+2q^{2}+2q^{\frac{5}{2}}+3q^{3}+4q^{\frac{7}{2}}+5q^{4}+7q^{\frac{9}{2}}+7q^{5}+O(q^{\frac{11}{2}})$.
Meanwhile  \begin{align*} & {\rm ch}[L_{1}^{N=2}](q)={\rm ch}[V_{\sqrt{3}\mathbb{Z}}](q)=\frac{\sum_{n\in\mathbb{Z}}q^{\frac{3}{2}n^{2}}}{\prod_{n \geq 1}(1-q^{n})}\\ &
=1+q+2q^{\frac{3}{2}}+2q^{2}+2q^{\frac{5}{2}}+3q^{3}+4q^{\frac{7}{2}}+5q^{4}+6q^{\frac{9}{2}}+7q^{5}+O(q^{\frac{11}{2}}).\end{align*}
Since in degree $\frac{9}{2}$ dimension of $A$ is bigger than the dimension of $V_{\sqrt{3}\mathbb{Z}}$ by $1$, the induced map \[\overline{\psi}:J_{\infty}(R_{L^{N=2}_1})/\langle a,b\rangle_{\partial}\rightarrow gr(L_{1}^{N=2})  \] is not injective.  
It is not hard to see that the one dimensional kernel of $\overline{\psi}$ in degree $\frac{9}{2}$ is spanned by \[c=G^{-}(-\frac{9}{2})-\frac{1}{3}h(-3)G^{-}(-\frac{3}{2})-G^{-}(-\frac{7}{2})h(-1)+\frac{1}{3}G^{-}(-\frac{5}{2})h(-1)^{2}.\] We make the following conjecture:

\begin{conj}
 The induced map $\hat{\psi}:J_{\infty}(R_{L_{1}^{N=2}})/\langle a,b,c\rangle_{\partial}\rightarrow gr(L_{1}^{N=2}) $ is an isomorphism.
\end{conj}

\section{principal subspaces}

Principal subspaces of affine vertex algebras (at least in a special case) were introduced by Feigin and Stoyanovsky \cite{feigin1993quasi} and further studied by several people; see \cite{butorac2012combinatorial, butorac2019principal, calinescu2008vertex, capparelli2006rogers, feigin2009principal,feigin2009principal} and references therein. In \cite{primc1994vertex,primc2000basic}, M.Primc studied Feigin-Stoyanovsky {\em type} subspaces  which are analogs of principal subspaces but easier to analyze. They are further investigated for many integral levels and types
\cite{baranovic2016bases,jerkovic2012character,trupvcevic2009combinatorial,trupvcevic2011characters}.
Here we follow notation from \cite{milas2012lattice} where principal subspaces are defined for general integral lattices (not necessarily positive definite).
As in \cite{milas2012lattice}, we let $V_L=M(1) \otimes \mathbb{C}[L]$ denote a lattice vertex algebra where $n={\rm rank}(L)$. We fix a $\mathbb{Z}$-basis $\mathcal{B}=\left\{\alpha_1,...,\alpha_n\right\}$ of $L$. Then the principal subspace associated to $\mathcal{B}$ and $L$, is defined as
$$W_L(\mathcal{B}):=\langle e^{\alpha_1},...,e^{\alpha_n} \rangle, $$
that is the smallest vertex algebra that contains extremal vectors $e^{\alpha_i}$. Once $\mathcal{B}$ is fixed, we shall drop the $\mathcal{B}$ in the parenthesis and write $W_{L}$ for brevity.

Let $\mathfrak{g}$ be a simple finite dimensional complex Lie algebra of type $A$, $D$ or $E$ and
let $\mathfrak{h}$ be a Cartan subalgebra of $\mathfrak{g}$. We choose simple roots $\left\{\al_{1},\ldots,\al_{n}\right\}$ and let $\Delta^{+}$ denote the set of positive roots.  Let $(\cdot,\cdot)$ be a rescaled Killing form on $\mathfrak{g}$ such that $(\al_{i},\al_{i})=2$ for $1\leq i\leq n$ (as usual we identify $\mathfrak{h}$ and $\mathfrak{h}^{*}$ via the Killing form). Fundamental weights of $\mathfrak{g}$, $\left\{\omega_{1},\ldots, \omega_{n}\right\}\subset \mathfrak{h}^{*}$, are defined by $(\omega_{i},\al_j)=\delta_{i,j}$ $(1\leq i,j \leq n)$.

Let $\mathfrak{n}_+$ be $\displaystyle \coprod_{\al\in \Delta^{+}}x_{\al}$, where $x_{\al}$ is the corresponding root vector and
$\widehat{\mathfrak{n}_{+}}=\mathfrak{n}_+ \otimes \mathbb{C}[t,t^{-1}]$ its affinization.
For an affine vertex algebra $L_{\widehat{\mathfrak{g}}}(k,0)$, $k \neq -h^\vee$, isomorphic to $L(k \Lambda_0)$ module,
we define the (FS)-principal subspace of simple $\widehat{\mathfrak{g}}$-module  $L_{\widehat{\mathfrak{g}}}(k,0)$ as
$$W_{\Lambda_{k,0}}:= U(\widehat{\mathfrak{n}}_+) \cdot \mathbb{1},$$ where $\mathbb{1}$ is the vacuum vector. It is easy to see that this is a vertex algebra (without conformal vector).
For $k=1$, we have $W_L \cong W_{\Lambda_{k,0}}$ where $L$ is the root lattice spanned by simple roots.

We fix a fundamental weight $\omega=\omega_m$ and  set 
 \[\Gamma=\left\{\al\in R|(\omega,\al)=1 \right\} \] where $R$ is the root system of $\mathfrak{g}$ and $\mathfrak{g}_{1}:=\coprod_{\al\in \Gamma}$ $\mathfrak{g}_{\al}$ where $\mathfrak{g}_{\al}$ is the root vector. This Lie algebra is commutative. We let $\mathfrak{g}_{1}\otimes \mathbb{C}[t,t^{-1}]$ be $\widehat{\mathfrak{g}_{1}}$. Then we can define the so-called Feigin-Stoyanovsky type subspace of $L_{\widehat{\mathfrak{g}}}(k,0)$ as \[W'_{\Lambda_{k,0}}:=U(\widehat{\mathfrak{g}}_1)\cdot v_{k \Lambda_0}.\]
 Unlike the FS subspace, this vertex subalgebra is commutative. We denote  $$\tilde{\Gamma}^{-}=\left\{x_{\gamma}(-r)|\gamma\in\Gamma,\; r>0 \right\}.$$
 $$\tilde{\Gamma}=\left\{x_{\gamma}(-r)|\gamma\in\Gamma,\; r\in \mathbb{Z} \right\}.$$ Notice that $U(\widehat{\mathfrak{g}}_1)\cong \mathbb{C}[\tilde{\Gamma}]$. Therefore  we can identify the elements in $W'_{\Lambda_{k,0}}$ with the elements in $\mathbb{C}[\tilde{\Gamma}^{-}].$
For any element in $W'_{\Lambda_{k,0}}$, \[v=x_{\beta_{1}}(m_{1})\ldots x_{\beta_{l}}(m_{l}),\quad  \beta_i\in \Gamma,\] we define colored weight as  $$cwt(v)=\displaystyle \sum_{i=1}^{l} \beta_{i}$$ for later use.

\subsection{Root lattices of ADE type} Following the notations in \cite{calinescu2008vertex}, we can prove the following
result.

\begin{prop}\label{slk}
For $\mathfrak{g}=sl(2)$, we have
 $W_{\Lambda_{k,0}}\cong gr (W_{\Lambda_{k,0}}) \cong J_{\infty}(\mathbb{C}[x]/\langle  x^{k+1}\rangle)$ for $k\geq 1$
\end{prop}
\begin{proof}
It is clear that $R_{W_{\Lambda_{k,0}}}=\mathbb{C}[x]/\langle  x^{k+1}\rangle.$ The result follows from Theorem 3.1 in \cite{calinescu2008vertex}.
\end{proof}

\begin{rem}
When $k=1$, $W_{\Lambda_{1,0}}$ is isomorphic to $J_{\infty}(\mathbb{C}[x]/\langle  x^{2}\rangle).$ By using different methods to calculate the Hilbert-Poincare series, see \cite{bruschek2013arc} and \cite{bai2018quadratic}, one can derive the famous Rogers-Ramanujan identities.
\end{rem}

For the rest of this subsection, we let $L$ be the $A_{n-1}$ root lattice with the rescaled Killing form $(\cdot,\cdot)$ such that $(\alpha,\alpha)=2$ for any root and the standard $\mathbb{Z}$-basis $\alpha_1,...,\alpha_{n-1}$ of simple roots. We are going to prove that $\psi$ is an isomorphism for the principal subspace $W_{L}$ corresponding to this basis. In the following, we will identify $W_{L}$ and $W_{\Lambda_{1,0}}$. Firstly we prove the following proposition:
\begin{prop}\label{prod}
Given elements $\alpha$, $\beta$, $\gamma$ and $\tau$ in lattice $L$, we have \begin{align} (e^\alpha)_{(-1)}e^{\beta}=0,\quad if\; (\alpha,\beta)>0 .\end{align}     \begin{align} (e^\alpha)_{(-1)}e^{\beta}=\frac{\epsilon(\al,\beta)}{\epsilon(\gamma,\tau)}(e^\gamma)_{(-1)}e^{\tau},\quad if\;\;\; (\alpha,\beta)=(\gamma,\tau)\;\;\; \text{and}\;\;\;\;\alpha+\beta=\gamma+\tau.\end{align}
\end{prop}
\begin{proof}
From the definition of vertex operators from \cite{kac1998vertex}, we have

 \[Y(e^{\alpha},z)e^{\beta}=\epsilon(\al,\beta)z^{(\alpha,\beta)}{\rm Exp}\left(\sum_{n\in \mathbb{Z}_{-}}\frac{-\alpha_{(n)}}{n}z^{-n} \right)e^{\alpha+\beta},\] where $\epsilon(\al,\beta)$ is a $2$-cocycle constant. Thus \[(e^\alpha)_{(-1)}e^{\beta}=\text{Coeff} _{z^{0}}(Y(e^{\alpha},z)e^{\beta})=0\] since the minimal power of $z$ above is greater than 0. The coefficients of $z^{0}$ of $Y(e^{\alpha},z)e^{\beta}$ and $Y(e^{\gamma},z)e^{\tau}$ are $\epsilon(\al,\beta)e^{\al+\beta}$ and $\epsilon(\gamma,\tau)e^{\gamma+\tau}$. The identity $(2)$ 
 follows from this fact and given condition.
\end{proof}

It is clear that all quotient relations in $R_{L_{\widehat{sl_{n}}}(1,0)}$ come from $(1)$ and $(2)$.  Thus $R_{L_{\widehat{sl_{n}}}}(1,0)\cap \mathbb{C}[E_{i,j}|1\leq i<j\leq n]=R_{W_{L}}.$

We be $E_{i,j}$ the $(i,j)$-th elementary matrix. Therefore $\left\{E_{i,j}\right\}_{1\leq i<j\leq n}$ is the set of all positive root vectors.  It is not hard to see that the $C_{2}$-algebra $R_{W_{L}}$ equals $\mathbb{C}[E_{i,j}|1\leq i<j\leq n]/I$ where we denote the equivalence class of $(E_{i,j})_{(-1)}\mathbf{1}$ by $E_{i,j}$.  In \cite[Corollary 2.7]{feigin2011zhu}, (see also \cite{feigin1993quasi} for $\mathfrak{g}=sl(3)$) authors have written down the graded decomposition of $R_{L_{\widehat{sl_{n}}}(1,0)}.$ By restricting it to its principal subspace, we have
\begin{prop}\label{c2p}
The $C_{2}$-algebra of $W_{L}$ equals \[ \mathbb{C}[E_{i,j}|1 \leq i<j \leq n]/\langle \sum_{\sigma \in S_{2}} E_{i_1,j_{\sigma_{1}}}E_{i_2,j_{\sigma_{2}}}|j_{1}>i_2 \rangle \] where $1\leq i_{1}\leq i_{2}\leq n$ and $1\leq j_{1}\leq j_{2}\leq n$.
\end{prop}

Moreover we have the following combinatorial $q$-identity which will be proven in a joint work with Milas \cite{Li.Milas}, where we also prove more general identities.

\begin{thm}[Li-Milas]\label{ML}
Let $A$ be the Cartan matrix $((\alpha_{i},\alpha_{j}))_{1\leq i,j\leq n-1}$ of type $A_{n-1}$, $n \geq 2$, and
$${\bf n}=(n_{1,2},....,n_{n-1,n})=(n_{i,j})_{1 \leq i < j  \leq n} .$$
Then we have
\begin{align}\sum_{{\bf n} \in \mathbb{N}_{\geq 0}^{n(n-1)/2}} \frac{{\displaystyle q^{B({\bf n})}}}{\displaystyle\prod_{1\leq i<j\leq n }(q)_{n_{i,j}}}=\displaystyle \sum_{\mathbf{k}=(k_{1},\ldots,k_{n-1})\in \mathbb{N}^{n-1}}\frac{q^{\mathbf{k}A\mathbf{k}^{\top}}}{(q)_{k_{1}}(q)_{k_{2}}\ldots (q)_{k_{n-1}}},\end{align}
where
$$B({\bf n})=\sum_{\substack{1\leq i_1 < j_1\leq n\\1\leq i_2 < j_2\leq n\\1\leq i_1 \leq i_2\leq n\\1\leq j_1 \leq j_2\leq n\\j_{1}>i_2 }} n_{i_{1},j_{1}}n_{i_{2},j_{2}}$$
\end{thm}

\begin{ex}
For $sl_{4},$ we have the following $q$-series identity:

   \begin{align*}&\displaystyle \sum_{\bf{n}\in\mathbb{N}_{\geq 0}^{6}}\frac{q^{n_{1}^{2}+n_{2}^{2}+n_{3}^{2}+n_{4}^{2}+n_{5}^{2}+n_{6}^{2}+n_1n_4+n_{1}n_{6}+n_2n_4+n_2n_5+n_3n_5+n_3n_6+n_4n_6+n_5n_6+a_4a_5}}{(q)_{n_{1}}(q)_{n_{2}}(q)_{n_{3}}(q)_{n_{4}}(q)_{n_{5}}(q)_{n_{6}}}\\ &= \sum_{\bf{k}\in \mathbb{N}^{3}_{\geq 0}} \frac{q^{k_{1}^{2}-k_{1}k_{2}+k_{2}^{2}-k_2k_3+k_3^2}}{(q)_{k_{1}}(q)_{k_{2}}(q)_{k_{3}}}\end{align*}   where we use multiindices $\mathbf{n}=(n_{1},n_{2},\ldots,n_{6})$ and $\mathbf{k}=(k_1,k_2,k_3).$ By doing following replacement, \[n_{1,2} \leftrightarrow n_1,\;\; n_{2,3} \leftrightarrow n_2,\;\; n_{3,4} \leftrightarrow n_3,\] \[n_{1,3} \leftrightarrow n_4,\;\; n_{2,4} \leftrightarrow n_5,\;\; n_{1,4} \leftrightarrow n_6,\] we recover the formula in Theorem \ref{ML}.

\end{ex}

Now we are ready to prove
\begin{thm}
The map $\psi$ is an isomorphism between $J_{\infty}(R_{W_{L}})$ and $gr(W_{L}).$
\end{thm}
\begin{proof}
From Proposition \ref{c2p}, we know that $J_{\infty}(R_{W_{L}})$ is isomorphic to \[ \mathbb{C}[E_{i,j}(n)|n\leq -1,\;1 \leq i<j \leq n]/\langle\sum_{\sigma \in S_{2}} E_{i_1,j_{\sigma_{1}}}(z)E_{i_2,j_{\sigma_{2}}}(z)|j_{1}>i_2 \rangle \] where $E_{i,j}(z)=\sum_{n\leq -1}E_{i,j}(n)z^{-n-1}$ and $1\leq i_{1}\leq i_{2}\leq n$, $1\leq j_{1}\leq j_{2}\leq n$. In order to simplify notation, we first order  $\left\{E_{i,j}\right\}_{1\leq i<j \leq n}$ as:
\begin{align*}
E_{1,2},E_{1,3},\ldots,E_{1,n},E_{2,3},\ldots,E_{2,n},\ldots, E_{n-1,n},
\end{align*}
and we denote this sequence by $\left\{E_{m}\right\}_{1\leq m \leq \frac{n(n-1)}{2}}$  (i.e. $E_1=E_{1,2}$, $E_2=E_{1,3}$ etc.). We then have a spanning set of jet algebra with each element of the form:
\begin{align*} E_{1}(-n_{1}^{1})\ldots E_{1}(-n_{1}^{k_{1}})E_{2}(-n_{2}^{1})\ldots E_{2}(-n_{2}^{k_{2}})\dots \end{align*}
 where $1\leq n_{m}^{k_{m}}\leq\ldots \leq n_{m}^{1}$ for $1\leq m\leq \frac{n(n-1)}{2}$. Here $k_{s}=0$ when we don't have terms involving $E_{s}$. Now we can reduce this spanning set by using quotient relations as following:
\begin{itemize}
    \item {\it difference two condition at distance 1}:
    If we have $E_{m}(z)^{2}=0$ in $C_{2}$-algebra, then we can impose a condition: $n_{m}^{p}\geq n_{m}^{p+1}+2$ $(1\leq p\leq k_{m}-1) $ on above spanning set.
    \item {\it boundary condition}:
    If we have $E_{s}(z)E_{t}(z)+\ldots=0$ $(s<t)$, we can impose a condition: $n_{s}^{k_{s}}\geq k_{t}+1.$
\end{itemize}
Therefore we have a reduced spanning set which implies \[HS_{q}(J_{\infty}(R_{W_{L}})) \leq  \sum_{{\bf n} \in \mathbb{N}_{\geq 0}^{n(n-1)/2}} \frac{{\displaystyle q^{B({\bf n})}}}{\displaystyle\prod_{1\leq i<j\leq n }(q)_{n_{i,j}}}.\] And it is well-known that \[{\rm ch}[gr(W_{L})](q)=\displaystyle \sum_{\mathbf{k}=(k_{1},\ldots,k_{n})\in \mathbb{N}^{n}}\frac{q^{\mathbf{k}A\mathbf{k}^{\top}}}{(q)_{k_{1}}(q)_{k_{2}}\ldots (q)_{k_{n}}}.\] Surjectivity of $\psi$ and identity $(2)$ together imply that $\psi$ is an isomorphism and the image of above spanning set under $\psi$ is a basis of $W_{L}.$
\end{proof}
\begin{rem}
Using result in \cite{milas2012lattice}, we can write down a basis of $W_{L}$ by using $(e^{\alpha_{i}})_{(j)}$ where $\alpha_{i}$ is a simple root of $sl_{n}$ and $j$ can be greater than or equal to $0$. If we want the subscript $j$ to be always less than $0,$ we have to include $(e^{\beta})_{(j)}$ where $\beta$ is a positive root. It is clear $E_{m}=E_{i_{m},j_{m}}$ is a root vector of a positive root $$\beta_m:=\al_{i_m}+\al_{i_m+1}+\ldots+\al_{j_m-1}.$$  Above proposition gives us a new basis of $W_{L}$,
\begin{align*}
(e^{\beta_1})_{(-n_{1}^{1})}\ldots (e^{\beta_1})_{(-n_{1}^{k_{1}})}(e^{\beta_2})_{(-n_{2}^{1})}\ldots (e^{\beta_2})_{(-n_{2}^{k_{2}})}\ldots (e^{\beta_M})_{(-n_{M}^{k_{M}})} \bf{1}
\end{align*}
where $M=\frac{n(n-1)}{2}$, $n_{M}^{k_{M}}\geq 1$, $n_{m}^{p}\geq n_{m}^{p+1}+2$ $(1\leq p\leq k_{m}-1)$ and $n_{s}^{k_{s}}\geq k_{t}+1$ if $1\leq s<t\leq M$, $i_t<j_{s}\leq j_{t}$.

\end{rem}

\vspace{0.2cm}

\subsection{Feigin-Stoyanovsky type subspaces}

\ \

In this section, we consider Feigin-Stoyanovsky type subspaces of affine vertex algebra of type $A_{n}$ at level $1.$  We first consider the special case when $\omega=\omega_{1}$. For any element of the $A_{n}$ root lattice,
$$\al=m_{1}\al_1+m_{2}\al_2+\ldots+m_{n}\al_n,$$ we define a subspace of $W'_{\Lambda_{1,0}}$ as $$(W'_{\Lambda_{1,0}})^{{\al}}:=\left\{v\in W'_{\Lambda_{1,0}}|\;cwt(v)=\al \right\}.$$ It is not hard to see that $(W'_{\Lambda_{1,0}})^{\al}$ is nontrivial if and only if \color{red}  \color{black} $m_1 \geq m_2\geq\ldots\geq m_n\geq 0$. According to \cite[(3.8)]{trupvcevic2011characters}, we have  \[{\rm ch}[(W'_{\Lambda_{1,0}})^{{\al}}](q)=\frac{q^{\sum_{i=1}^{n}m_i^2-\sum_{i=1}^{n-1}m_im_{i+1}}}{(q)_{m_n}(q)_{m_{n-1}-m_n}\ldots (q)_{m_1-m_2}}. \] Then \begin{align*} {\rm ch}[W'_{\Lambda_{1,0}}](q)&=\displaystyle \sum_{0\leq m_{n}\leq\ldots\leq m_1}\frac{q^{\sum_{i=1}^{n}m_i^2-\sum_{i=1}^{n-1}m_im_{i+1}}}{(q)_{m_n}(q)_{m_{n-1}-m_n}\ldots (q)_{m_1-m_2}}\\ &=\sum_{(l_{1},\dots,l_{n})\in\mathbb{N}^{n}}\frac{q^{\sum_i^n l_i^2+\sum_{1\leq i<j\leq n}{l_{i}l_{j}}}}{(q)_{l_1}(q)_{l_2}\ldots (q)_{l_n}}. \end{align*}
Moreover, in this case, \[\Gamma=\left\{\beta_1:=\al_1,\beta_2:=\al_1+\al_2,\ldots,\beta_n:=\al_1+\ldots+\al_n\right\}.\]
Notice that $$L=\mathbb{Z}\beta_1\oplus\ldots\oplus \mathbb{Z}\beta_n$$ is a lattice with basis $\left\{\beta_1,\ldots,\beta_n\right\}$. Then we have \[W_{L}\cong W'_{\Lambda_{1,0}}.\] It is not hard to see that \begin{align*}&\langle \beta_i,\beta_i \rangle=2,\quad if\; 1\leq i\leq n \\ &\langle \beta_i,\beta_j \rangle=1,\quad if\; 1\leq i\neq j\leq n.\\ \end{align*}
According to Proposition \ref{prod}, we have that $C_{2}$-algebra of $W_L$ is $$\mathbb{C}[x_1,\ldots,x_n]/\langle x_{i}x_{j}|1\leq i\leq j\leq n\rangle.$$
By similar argument in previous section, we get \[HS_{q}(J_{\infty}(\mathbb{C}[x_1,\ldots,x_n]/\langle x_{i}x_{j}|1\leq i\leq j\leq n\rangle)])={\rm ch}[W_{{L}}](q)\] which implies isomorphism between $J_{\infty}(R_{W'_{\Lambda_{1,0}}})$ and $gr(W'_{\Lambda_{1,0}}).$ Similarly we can also prove the isomorphism in cases where $\omega=\omega_{i}$, $2\leq i\leq n$ by making use of \cite[(3.21)]{trupvcevic2011characters}.

\subsection{Principal subspaces and jet schemes  from graphs}

In this part we study principal subspaces and jet algebras coming from graphs. 
We begin from any graph $G$ with $k$ vertices and possibly with loops (and for simplicity we assume no double edges). We denote the vertices of $G$ by $\left\{v_1,v_2,\ldots,v_k\right\}$. We denote by $\Gamma:=\Gamma(G)$ the (symmetric) incidence matrix of $G$ and by $(L(\Gamma), \langle \  ,  \ \rangle )$ rank $k$  lattice with
basis $\alpha_1,...,\alpha_k$, such that $\langle \alpha_i,\alpha_j \rangle=(\Gamma)_{i,j}$.
The incidence matrix of the graph induces a quadratic form
$$\Gamma \rightarrow \frac{1}{2}Q(x_1,...,x_k),$$
where
$$Q(x_1,...,x_k)= \sum_{\substack{i,j=1 \\ v_i v_j \in E(G)}} ^k x_i x_j$$
where we sum over all edges $E(G)$.
Out of monomials appearing in the sum we form the infinite jet scheme  $J_\infty X(\Gamma)$ where $$R_\Gamma=\mathbb{C}[x_1,\ldots,x_k]/\langle \cup_{ v_i v_j \in E(G)} x_i x_j \rangle.$$
We let $W_{L(\Gamma)} \subset V_{L(\Gamma)}$ be the principal subspace corresponding to $\{e^{\alpha_i} \}_{i=1}^k$
inside the lattice vertex algebra $V_{L(\Gamma)}$. For simplicity we write $W_\Gamma$ for $W_{L(\Gamma)}$.

\begin{ex}
Consider the graph $\circ-\circ-\circ$. Then $\Gamma=\left[ \begin{array}{ccc} 0 & 1 & 0 \\ 1 & 0 & 1 \\ 0 & 1 & 0 \end{array}\right]$, and $W_{\Gamma}=\langle e^{\alpha_1},e^{\alpha_2}, e^{\alpha_3} \rangle$ where
$L=\mathbb{Z}\alpha_1 \oplus \mathbb{Z} \alpha_2 \oplus \mathbb{Z} \alpha_3$ with $\langle \alpha_1,\alpha_2 \rangle =\langle \alpha_2,\alpha_3 \rangle =1$ (zero otherwise),
$R_\Gamma=\mathbb{C}[x_1,x_2,x_3]/(x_1 x_2, x_2 x_3),$  and  $Q(x_1,x_2,x_3)=x_1 x_2+x_2 x_3$.
\end{ex}

\begin{thm} \label{graph}
If the bilinear form associated with $\Gamma$ is non-degenerate, that is $\Gamma$ is invertible, then there exists a unique conformal vector in lattice vertex algebra such that eigenvalue of $L_{(0)}$ defines grading such that:

\[wt(e^{\alpha_{i}})=\frac{3}{2} \quad \text{if} \quad \langle  \alpha_{i},\alpha_{i}\rangle=1,\]
\[ wt(e^{\alpha_{i}})=1 \quad \text{if} \quad \langle  \alpha_{i},\alpha_{i}\rangle=0.\] Moreover, the graded dimension is given by:

$${\rm ch}[W_{\Gamma}](q)=\sum_{n_{1},\ldots,n_{k}\geq 0}\frac{q^{n_{1}+n_{2}+\ldots+n_{k}+\frac{1}{2}Q(n_{1},\ldots,n_{k})}}{(q)_{n_{1}}\ldots(q)_{n_{k}}}.$$\end{thm}
\begin{proof}

Clearly, we have the standard conformal vector in lattice vertex algebra given by $\omega_{st}=\frac{1}{2}\sum_{i=1}^{n}u^{(i)}_{(-1)}u^{(i)}_{(-1)}\mathbf{1}$ where $u^{(1)},\ldots,u^{(n)}$ is an orthonormal basis with respect to the bilinear form associated with $\Gamma.$ We know that $$L_{st}(0) (e^{\alpha_{i}})=\frac{\langle \alpha_{i},\alpha_{i}\rangle}{2}.$$ It is clear that by adding a linear combination of $\left\{(\alpha_{i})_{(-2)}\mathbf{1}\right\}_{i=1}^{n}$, we will still get a conformal vector. Now assume that $\omega_{st}+\sum_{i=1}^{n}a_{i}(\alpha_{i})_{(-2)}\mathbf{1}$ where $a_{i}\in\mathbb{C}$ would give us expected weights. Then we have a system of linear equations. The non-degeneracy of the bilinear form implies that there is an unique solutions set. Thus we always have a conformal vector with the grading:$$wt(e^{\alpha_{i}})=\frac{3}{2} \quad \text{if} \quad \langle  \alpha_{i},\alpha_{i}\rangle=1,$$
$$ wt(e^{\alpha_{i}})=1 \quad \text{if} \quad \langle  \alpha_{i},\alpha_{i}\rangle=0.$$

By applying \cite[Corollary 4.14]{milas2012lattice}, we can write a combinatorial basis of $W_{\Gamma}.$ 
Now let us use this basis to write down the character. Firstly, the generating function of colored partition into $(n_{1},n_{2},\ldots,n_{k})$ parts is $\frac{1}{(q)_{n_{1}}\ldots(q)_{n_{k}}}.$ It is clear that
\[{\rm ch}[W_{\Gamma}](q)=\sum_{k_{1},\ldots,k_{k}\geq 0}\frac{q^{wt(f_{(n_1,\ldots,n_k)})}}{(q)_{n_{1}}\ldots(q)_{n_{k}}},\] where $f_{(n_1,\ldots,n_k)}$ is the vector in $W_{\Gamma}$ of charge $(n_{1},\ldots,n_{k})$ with the minimal weight.  For the $n_{i}$ part, there is an unique element $u_{n_{i}}$ of the minimal weight which is \[e^{\al_i}_{(-1-\sum_{j=1}^{i-1}\langle  \alpha_{i},\alpha_{j} \rangle n_{j}-(n_i-1)\langle \al_i,\al_i \rangle)}\ldots e^{\al_i}_{(-1-\sum_{j=1}^{i-1}\langle  \alpha_{i},\alpha_{j} \rangle n_{j})}\bf{1}.\]  The weight of $u_{n_{i}}$ is \begin{align*}&\frac{n_{i}}{2}(2(\sum_{j=1}^{i-1}\langle  \alpha_{i},\alpha_{j} \rangle n_{j}+wt((e^{\alpha_{i}})_{(-1)}\textbf{1})+(n_{i}-1)\langle  \alpha_{i},\alpha_{i}\rangle)\\ &=\sum_{j=1}^{i-1}\langle  \alpha_{i},\alpha_{j}\rangle n_{i}n_{j}+\frac{n_{i}^{2}}{2}\langle  \alpha_{i},\alpha_{i}\rangle+(-\frac{\langle  \alpha_{i},\alpha_{i}\rangle}{2}+wt((e^{\alpha_{i}})_{(-1)}\textbf{1}))n_{i}. \end{align*}
Therefore \begin{align*} &wt(f_{(n_1,\ldots,n_k)})=\sum_{i=1}^{k} wt(u_{n_{i}})\\
    &=\sum_{i=1}^{k}\sum_{j=1}^{i-1}\langle  \alpha_{i},\alpha_{j}\rangle n_{i}n_{j}+\frac{n_{i}^{2}}{2}\langle  \alpha_{i},\alpha_{i}\rangle+(-\frac{\langle  \alpha_{i},\alpha_{i}\rangle}{2}+wt((e^{\alpha_{i}})_{(-1)}\textbf{1}))n_{i}\\&=n_{1}+n_{2}+\ldots+n_{k}+\frac{1}{2}Q(n_{1},\ldots,n_{k}).
\end{align*}
Thus we proved the claimed identity.
\end{proof}
\begin{rem}
If the lattice $L$ is degenerate, then $V_L$ has no conformal vector which can give us expected weights. But we can still view $W_L$ as a graded vertex algebra if we define the degree of $e^{\alpha_{i}}$ as above. Then the character formula is still valid for singular $\Gamma$.
\end{rem}

Before we prove next result, let us generalize \cite[Theorem 4.3.1]{penn2014lattice}.
\begin{prop}\label{pre} We have an isomorphism
\begin{align*}
   &  gr(W_{\Gamma}) \\
    & \cong \frac{\mathbb{C}[x_{i}(p)|1\leq i\leq k, p\leq -1]}{ \langle  \sum_{m=0}^{-l-1}\frac{(m+\langle  \alpha_{i},\alpha_{j}\rangle-1)!}{m!}\langle  \alpha_{i},\alpha_{j}\rangle x_{i}({-\langle  \alpha_{i},\alpha_{j}\rangle-m})x_{j}(l+m)|1\leq i,j\leq k,l\leq -1\rangle.} \end{align*}
\end{prop}
\begin{proof}
First, we define a map $\pi$ from \begin{align*}
      \mathbb{C}[x_{i}(p)|1\leq i\leq k, p\leq -1] \end{align*} to $gr(W_{\Gamma})$ by sending $x_{i}(p)$ to $e^{\alpha_{i}}_{(p)}\bf{1}$.  We denote the ideal$$\langle  \sum_{m=0}^{-l-1}\frac{(m+\langle  \alpha_{i},\alpha_{j}\rangle-1)!}{m!}\langle  \alpha_{i},\alpha_{j}\rangle x_{i}(-\langle  \alpha_{i},\alpha_{j}\rangle-m)x_{j}(l+m)|1\leq i,j\leq k,l\leq -1\rangle$$ by $I_{\Gamma}$. Next we use an argument from \cite{penn2014lattice} to show that $I_{\Gamma}\subset ker(\pi)$.

      We prove $ker(\pi)\subset I_{\Gamma}$ by contradiction. Suppose there exists an element $a\in  \mathbb{C}[x_{i}(p)|1\leq i\leq k, p\leq -1]$ such that $a\in ker(\pi)$ and $a\notin I_{\Gamma}.$ Suppose $a$ is homogeneous with respect to weight and charge. Choose $r$ such that $a$ contains some element $x_{r}(p)$ as a factor. We assume that $a$ has the minimum weight among all elements that satisfy above conditions. Again from the same argument from \cite{penn2014lattice}, this $a$ can be written as $ux_{r}(-1) $ where $u\in \mathbb{C}[x_{i}(p)|1\leq i\leq k, p\leq -1]$. We prove the case when $\langle \alpha,\alpha \rangle=0.$ For other cases, it is proved in \cite{penn2014lattice}. Firstly we define a map ${\bf e^{\alpha_{r}}}:W_{\Gamma}\rightarrow W_{\Gamma}$ as $${\bf e^{\alpha_{r}}}((e^{\alpha_{j}})_{(m)}\textbf{1})=(e^{\alpha_{j}})_{(m)}(e^{\alpha_{r}})_{(-1)}\textbf{1}.$$ Then we lift this map to $${\bf x_{r}}: \mathbb{C}[x_{i}(p)|1\leq i\leq k, p\leq -1]\rightarrow  \mathbb{C}[x_{i}(p)|1\leq i\leq k, p\leq -1] $$ which is defined as $${\bf x_{r}}(x_{i}(j))=x_{i}(j)x_{r}(-1).$$
      Since $a\in ker(\pi),$ $\pi(a)=\pi(bx_{r}(-1))=\pi(b)(e^{r})_{(-1)}\mathbf{1}=0.$ Then $${\bf {e^{\alpha_{r}}}}(\pi(b)(e^{r})_{(-1)}\textbf{1})=\pi(b)\textbf{1}=0 $$ which implies that $b\in ker(\pi).$ If $b\in I_{\Gamma},$ then $a={\bf x_{r}}(b)\in {\bf x_{r}}I_{\Gamma}\in I_{\Gamma}$ which contradicts with our assumption. If $b\notin I_{\Gamma},$ then $b$ is an element such that $b\in ker(\pi)$ and $b\notin I_{\Gamma}$ but with the weight strictly less than the weight of $a$. This also contradicts our assumption. Thus we proved the claim.
\end{proof}

\begin{thm}\label{prin} We have that $$gr(W_{\Gamma})\cong J_{\infty} (\mathbb{C}[y_{1},y_{2},\ldots, y_{k}]/\langle \langle   \alpha_{i},\alpha_{j}\rangle y_{i}y_{j}|1\leq i,j\leq k\rangle).$$
\end{thm}
\begin{proof}

   From the definition of jet superalgebra,  we know that $$T^{(-l-1)}(\langle  \alpha_{i},\alpha_{j}\rangle y_{i}y_{j})=\sum_{m=0}^{-l-1}c_{m}^{l}\langle  \alpha_{i},\alpha_{j}\rangle y_{i}(-\langle  \alpha_{i},\alpha_{j}\rangle -m)y_{j}(l+m)$$ where $c_{m}^{l}$ is a constant coefficient. Therefore $$J_{\infty}(\mathbb{C}[y_{1},y_{2},\ldots,y_{k}]/\langle  \langle  \alpha_{i},\alpha_{j}\rangle y_{i}y_{j}|1\leq i,j\leq k\rangle )$$ has quotient relation $$\langle  \sum_{m=0}^{-l-1}c_{m}^{l}\langle  \alpha_{i},\alpha_{j}\rangle y_{i}(-\langle  \alpha_{i},\alpha_{j}\rangle -m)y_{j}(l+m)|1\leq i,j\leq k,l\leq -1\rangle .$$ Together with Proposition \ref{pre}, we get an isomorphism of differential algebras induced from the map $\psi:x_{i}(-1)\rightarrow y_{i}(-1)$.
\end{proof}
When $\langle  \alpha_{i},\alpha_{i}\rangle =1$, we increase the degree of $y_{i}(-1) $ by $\frac{1}{2}$.
Then clearly we have
$$HS_{q}(J_{\infty}(R_\Gamma))={\rm ch}[W_{\Gamma}](q).$$

\subsection{Positive lattices}
\color{red}
\color{black}

Given a lattice $L$ of rank $n$ with a $\mathbb{Z}$-basis $\left\{\al_i\right\}_{i=1}^{n}$,
We say that the basis is positive if we have $\langle \alpha_{i},\alpha_{j}\rangle\geq 0$ for $1\leq i\leq j\leq n$. In this part, we study principal subspaces associated with positive bases. Examples we studied in previous two sections are such principal subspaces. Now let us prove a more general result about the map $\psi$ and such principal subspaces.

\begin{thm}
For a lattice $L$ of rank $n$ with a positive basis, the map $\psi$ is an isomorphism for $W_{L}$ if and only if its positive basis satisfies $\langle \alpha_{i},\alpha_{i}\rangle=a$, where $a=0$ or $1$ or $2$ and $\langle \alpha_{i},\alpha_{j}\rangle=b$, where $b=0$ or $1$.
\end{thm}

\begin{proof}

First let us assume that the positive basis of the lattice $L$ satisfies given conditions. Notice that according to Theorem \ref{prin}, we know that when $\langle \alpha_{i},\alpha_{i}\rangle=a$, where $a=0$ or $1$ and $\langle \alpha_{i},\alpha_{j}\rangle=b$, where $b=0$ or $1$, the map $\psi$ is an isomorphism for the principal subspace. Now the only case we need consider is the positive basis for which $\langle \alpha_{i},\alpha_{j}\rangle=2\delta_{i,j}.$ It is not hard to see that $J_{\infty}(\mathbb{C}[x]/x^{2})$ has a basis
$$\left\{x_{(m_{1})}x_{(m_{2})}\dots x_{(m_{k})}| m_{j-1}\leq m_{j}-2, k\geq 0 \right\}.$$  Thus $J_{\infty}(\mathbb{C}[x_{1},x_{2}, \ldots, x_{n}]/\langle  x_{1}^{2},x_{2}^{2}, \ldots, x_{n}^{2}\rangle)$ has a basis
\begin{align*}&\{(x_{i_{1}})_{(m^{1}_{1})}(x_{i_{1}})_{(m^{1}_{2})} \dots (x_{i_{1}})_{(m^{1}_{k_{1}})} \ldots (x_{i_{n}})_{(m^{n}_{1})}(x_{i_{n}})_{(m^{n}_{2})} \dots (x_{i_{n}})_{(m^{n}_{k_{n}})} \\ &| m^{i}_{j-1}\leq m^{i}_{j}-2, \;  1\leq j\leq k_{i}-1 \}.\end{align*} Notice that the $C_{2}$-algebra of $W_{L}$ is
$$\mathbb{C}[x_{1},\ldots,x_{n}]/\langle x_{1}^{2},\ldots,x_{{n}}^{2}\rangle$$ Now the map $\psi$  is sending $(x_{i})_{(-1)}$ to $(e^{\alpha_i})_{(-1)}{\bf 1}$. According to \cite[Corollary 4.14]{milas2012lattice}, the image of the basis of $J_{\infty}(R_{W_{L}})$ is the basis of $gr(W_{L})$. Thus the map $\psi$ is an isomorphism.

Next, let us prove that if the basis does not satisfy given conditions, the map $\psi$ is not an isomorphism. We will consider two cases:

\begin{itemize}
    \item Suppose that for one simple root $\alpha_{i}$, we have $\langle \alpha_{i},\alpha_{i}\rangle \geq 3.$ Without loss generality, we prove that $\psi$ is not an isomorphism when lattice $L=\mathbb{Z}\alpha_{i}.$  In this case, from  \cite[Corollary 4.14]{milas2012lattice}, the basis of $gr^{F}W_{L}$ is \[\left\{(e^{\alpha_{i}})_{(m^{1})}(e^{\alpha_{i}})_{(m^{2})} \dots (e^{\alpha_{i}})_{(m^{k})} {\bf 1} | m_{j-1}\leq m_{j}-\langle \alpha_{i},\alpha_{i}\rangle, \;  m_{k}< 0,\;k\geq 0 \right\}.\] It is easy to see that neither $J_{\infty}(\mathbb{C}[x]/(x^{2}))$ nor $J_{\infty}(\bigwedge[x])$ has the same corresponding basis (here $\bigwedge$ denotes the exterior algebra).
   \item Suppose that there exists two distinct roots $\alpha_{i},\alpha_{j}$ where $i<j$ such that $\langle \alpha_{i},\alpha_{j}\rangle\geq 2$. Without loss of generality, we assume $L=\mathbb{Z}\al_i\oplus \mathbb{Z}\al_j$, then the basis of $J_{\infty}(W_L)$ is
    \begin{align*}&\{(x_{i})_{(-1-m_{1})}(x_{i})_{(-1-m_{2})}\ldots (x_{i})_{(-1-m_{k})}(x_{j})_{(-1-n_{1})}(x_{j})_{(-1-n_{2})}\ldots (x_{j})_{(-1-n_{l})}\\&|m_1-m_2\geq \langle \al_i,\al_i \rangle ,\;n_1-n_2\geq\langle \al_j,\al_j \rangle,\; m_k\geq l,n_l\geq 0 \}. \end{align*} Meanwhile according to \cite[Corollary 4.14]{milas2012lattice}, the image of this basis under $\psi$ strictly contains the basis of $W_{L}$. We do not have isomorphism.
\end{itemize}
Thus we proved the statement.
\end{proof}

\subsection{New character formulas for \texorpdfstring{${\rm ch}[W_{\Gamma}]$}{Lg}}
Here we continue from Section 5.3. If the graph $\Gamma$ is of Dynkin type $A_k$ - path of length $k-1$) or $C_k$ (cycle of length $k$)
we expect that the generating series $HS_{q}(J_{\infty}(R_{\Gamma}))$ has much better behaved combinatorial and perhaps even mock modular properties. We now present "sum of tails" formulas for $HS_{q}(J_{\infty}(R_{A_k}))$ for
several low "rank" cases. To simplify notation we let $$A_k(q):=HS_{q}(J_{\infty}(R_{A_k})).$$
From Theorem \ref{graph} we have a fermionic formula
\begin{equation} \label{master}
 A_k(q)=   \sum_{n_1,n_2,\dotsc,n_k\ge0}
\frac{q^{n_1+n_2+\dotsb+n_k+n_1n_2+n_2n_3+\dotsb+n_{k-1}n_k}}
{(q)_{n_1}^{} (q)_{n_2}^{} \dotsm (q)_{n_k}^{} },
\end{equation}

Next formulas are recently given by Jennings-Shaffer and Milas \cite{JM}.
\begin{thm}  \label{JM}
We have
\begin{itemize}
\item

\vspace{.2cm}

$\dis{A_2(q)=
	\frac{1}{(1-q)(q)_\infty}}$

\item
\vspace{.2cm}

$\dis{A_3(q)=q^{-1} \left( \frac{1}{(q)^2_\infty}-
		\frac{1}{(q)_\infty} \right)}$

\item

\vspace{.2cm}
$\dis{A_4(q)=\frac{q^{-1}}{(q)_\infty^2}\sum_{n\ge1}\frac{q^n}{1-q^n}}$

\item

\vspace{.2cm}

$\dis{A_5(q)=\frac{1}{(q)_\infty^2}\sum_{n\ge0}\frac{q^n}{(q)_{n}(1-q^{n+1})^2}}.$

\item

$\dis{A_6(q)=\frac{1}{(q)_\infty^2}\sum_{n,m\ge0}\frac{q^{n+m+nm}}{(q)_{n+1}(q)_{m+1}}}.$

\end{itemize}

\end{thm}

Moreover, for cyclic graphs $C_k$-graphs we have fermionic formulas for  $C_k(q):=HS_{q}(J_{\infty}(R_{C_k}))$ valid for $k \geq 3$
\begin{equation} \label{master2}
C_k(q)=\sum_{n_1,n_2,\dotsc,n_k\ge0}
\frac{q^{n_1+n_2+\dotsb+n_k+n_1n_2+n_2n_3+\dotsb+n_{k-1}n_k+n_k n_1}}
{(q)_{n_1}^{} (q)_{n_2}^{} \dotsm (q)_{n_k}^{} }.
\end{equation}
Again we have partial results for "bosonic" representations for $3-$ and $5$-cycle graphs \cite{JM}.
\begin{prop} \label{JM2} We have

\vspace{0.2cm}

$\dis{C_3(q)=\frac{1}{(q)_\infty} \sum_{n \geq 0} \frac{q^n}{ (q^{n+1})_{n+1} }}.$

\vspace{0.2cm}

$\dis{C_5(q)=\frac{q^{-1}}{(q)_\infty^2} \sum_{n \geq 1}\frac{nq^n}{1-q^{n}}}.$

\end{prop}

\subsection{Combinatorial interpretation}

Next we present  combinatorial interpretations of formulas in Theorem \ref{JM} and Proposition \ref{JM2}. For simplicity, in several
formulas we factored out a (power of)  Euler factor which can be easily interpreted as the number of
(colored) partitions.

\begin{thm} We have:
\begin{itemize}
\item $A_2(q)$ counts the number of partitions of $2n$ with all parts either even or equal to 1.
\item $q A_3(q)$ counts the number of partitions of $n+1$ into two kinds of parts with the first kind of parts used in each
partition.
\item $q (q)_\infty A_4(q)$ counts the total number of parts in all partitions of $n$, which is also sum of largest parts of all
partitions of $n$.
\item $(q)_\infty^2 A_5(q)$  is the sum of the numbers of times that the largest part appears in each partition of $n$.
\item $q (q)_\infty^2  A_6(q)$ counts twice the total number of parts in all partitions of $n$ minus the number of partitions of $n$.
\item $(q)_\infty C_3(q)$ counts the number of partitions of n such that twice the least part is bigger than the greatest part.
\item $q (q)_\infty C_5(q)$ counts the sum of all parts of all partitions of $n$, also known as ${\rm n p}(n)$.
\end{itemize}

\end{thm}
\begin{proof}

For $A_2(q)$, observe that ${\rm Coeff}_{q^{n}}A_2(q)=p(1)+p(2)+\cdots + p(n)$, where $p(i)$ is the number of partitions of $i$. The number of $1's$ must be even, say $2k$, so we have to compute the number of
partitions of $2n-2k$ where all parts are even. This is given by $p(n-k)$. Then summing over $k$ gives the claim.

The interpretation for the $A_3(q)$ series, is clear because we can also write $$q A_3(q)=\frac{1}{(q)_\infty}
	\left(
		\frac{1}{(q)_\infty}-1
	\right).$$
Extracting the coefficient on the right-hand side gives $p_2(n)-p(n)$, where $p_2(i)$ denotes the number of two colored partitions.

For $A_4(q)$, this can be seen from identity
$\frac{\sum_{n \geq 1} \frac{q^n}{1-q^n}}{(q)_\infty}=\sum_{n \geq 1} \frac{nq^n}{(q)_n} ,$
which follows by taking the $(x \frac{d}{d})$ derivative of $\frac{1}{(xq;q)_\infty}=\sum_{n \geq 0} \frac{x^n q^n}{(q)_n}$.
This clearly counts the total number of parts in all partitions of $n$.

The $(q)_\infty^2 A_5(q)$ case was already discussed in  \cite{JM}.

For $(q)_\infty^2 A_6(q)$, this follows from another identity given in \cite{JM}:
$$\frac{1}{(q)_\infty^2}\sum_{n,m\ge0}\frac{q^{n+m+nm}}{(q)_{n+1}(q)_{m+1}}=\frac{q^{-1}}{(q)^2_\infty} \left(2 \sum_{n \geq 1} \frac{q^n}{(1-q^n)(q)_\infty}+1-\frac{1}{(q)_\infty} \right),$$
together with a previous observation that
$ \frac{\sum_{n \geq 1}\frac{q^n}{1-q^n}}{(q)_\infty}$ counts the total number of parts in all partitions of $n$.

For $(q)_\infty C_3(q)$ we use a well-known interpretation for the fifth order mock theta function, and finally
for $(q)_\infty C_5(q)$ we observe the formula $$\left(q \frac{d}{dq} \right) \frac{1}{(q)_\infty}=\frac{1}{(q)_\infty} \sum_{n \geq 1}\frac{nq^n}{1-q^{n}}=\sum_{n \geq 1} n p(n)q^n$$ as claimed.

\end{proof}

\begin{rem} It is interesting to observe that the numerators of  $C_3(q)$ and $C_5(q)$ are mock modular forms, and thus
$C_3(q)$ and $C_5(q)$ are {\em mixed} mock.
Completion of the Ramanujan fifth order mock theta   $\sum_{n \geq 0} \frac{q^n}{ (q^{n+1})_{n+1}} $
function is well-documented \cite{BFOR} . For $\sum_{n \geq 1} \frac{nq^n}{1-q^n}$ we only have to observe that adding $-\frac{1}{24}$ to the numerator gives $E_2(\tau)$,  the weight two quasimodular Eisenstein series, which is known to be mock.
\end{rem}
\section{\texorpdfstring{$N=1$}{Lg}  superconformal vertex algebra}

In this section we consider the rational $N=1$ vertex superalgebra $L_{c_{2,4k}}^{N=1}$ $(k\geq 1)$ associated to $N=1$ superconformal $(2,4k)$-minimal models \cite{adamovic1997rationality}. Here the central charge is $c_{2,4k}=\frac{3}{2}(1-\frac{2(4k-1)^{2}}{8k})$.

 According to \cite{melzer1994supersymmetric,milas2007characters}, we know that the normalized character of $L_{c_{2,4k}}^{N=1}$ (without the $q^{-c/24}$ factor) is:
\begin{align*} {\rm ch}[L_{c_{2,4k}}^{N=1} ](q)& =\displaystyle\prod^{\infty}_{\substack{n=1\\n\not \equiv 2(\text{mod}\;2)\\n\not \equiv 0,\;\pm 1 (\text{mod}\;4k)}} \frac{1}{(1-q^{\frac{n}{2}})}  \\&=\displaystyle\sum_{m_{1},\ldots,m_{k-1} \geq 0} \frac{(-q^{\frac{1}{2}})_{N_{1}}q^{\frac{1}{2}N_{1}^{2}+N_{2}^{2}+\ldots+N_{k+1}^{2}+N_{(s+1)/2}+N_{(s+3)/2}+\ldots+N_{k-1}}}{(q)_{m_{1}}(q)_{m_{2}}\ldots (q)_{m_{k-1}}}.\end{align*} And the fermionic character formula is the generating function
(cf. \cite{melzer1994supersymmetric})
 $${\rm ch}[L_{c_{2,4k}}^{N=1}](q)=\sum_{n=0}^{\infty}D_{k,1}(n)q^{\frac{n}{2}}$$ of the number of partitions of $D_{k,1}(n)$ of $\frac{n}{2}$ in the form $\frac{n}{2}=b_{1}+\ldots+b_{m}$ $(b_{j}\in \frac{1}{2}\mathbb{Z}_{\geq 1})$ where  $b_{1},\ldots, b_{m}$ satisfy the following conditions:

\begin{itemize}
    \item no half-odd integer is repeated.
    \item $b_{j}\geq b_{j+1},$ $b_{m}\geq \frac{3}{2}$,
    \item $b_{j}-b_{j+k-1}\geq 1$ if $b_{j}\in \mathbb{Z}+\frac{1}{2},$
    \item $b_{j}-b_{j+k-1}> 1$ if $b_{j}\in \mathbb{Z}.$
\end{itemize}
Since $N=1$ vertex superalgebra $L_{c_{2,4}}^{N=1}$ is isomorphic to $\mathbb{C}$, we only need consider $L_{c_{2,4k}}^{N=1}$ where $k>1$. First let us find the $C_2$-algebra of $L_{c_{2,4k}}^{N=1}$. According to \cite[Section 4]{milas2007characters}, we know that the null vector in universal algebra which survives inside the $C_{2}$-algebra is $L_{(-2)}^{k-1}G_{(-\frac{3}{2})}\mathbf{1}.$ Moreover if we let $G_{(-\frac{1}{2})}$ act on the null vector, we get another null vector which survives in $C_{2}$-algebra, i.e $L_{(-2)}^{k}\mathbf{1}.$ These two null vectors in the vacuum algebra  generate the whole quotient ideal of $R_{L_{c_{2,4k}}^{N=1}}$. Thus $R_{L_{c_{2,4k}}^{N=1}}$ is isomorphic to superalgebra $\mathbb{C}[l,g]/\langle  l^{k},l^{k-1}g\rangle $ where $g$ is an odd element.

We are going to prove that $\psi$ is an isomorphism. We identify $l$, $g$ with $l(-2)$, $g(-\frac{3}{2})$, respectively, inside the jet superalgebra.

It is clear that $J_{\infty}(\mathbb{C}[l,g]/\langle  l^{k},l^{k-1}g\rangle )$ is isomorphic to
$$\mathbb{C}[l(-2-i),g(-\frac{3}{2}-j)|i,j\geq 0]/\langle  l(z)^{k},l(z)^{k-1}g(z)\rangle $$ where $l(z)=\displaystyle \sum_{n\in \mathbb{N}}l(-2-n)z^{n}$, $g(z)=\displaystyle \sum_{n\in \mathbb{N}}g(-\frac{3}{2}-n)z^{n}$ and $\langle  l(z)^{k},l(z)^{k-1}g(z)\rangle$ is the ideal generated by the Fourier coefficients of $l(z)^{k},l(z)^{k-1}g(z)$.  We define {\em ordered monomial} in $J_{\infty}(\mathbb{C}[l,g]/\langle  l^{k},l^{k-1}g\rangle )$ to be a monomial of the form $$ l(-2-n)^{a_{1}}g(-\frac{3}{2}-n))^{b_{1}}l(-1-n)^{a_{2}}g(-\frac{3}{2}-n+1))^{b_{2}}\ldots l(-2)^{a_{n+1}}g(-\frac{3}{2})^{b_{n+1}}$$ where $n\geq 0.$ Then we have a complete lexicographic ordering on all ordered monomials according to Section \ref{ordering}.

We know that all ordered monomials constitute a spanning set of the jet superalgebra. Following the similar argument in Section \ref{N2}, we can make use of the quotient relation to impose some conditions on the spanning set to get a smaller spanning set. Firstly since all variables $g(k)'s$ are odd, no two $g(k)$ can appear in the ordered monomial. The leading term of any coefficient of $z^{nk}$ in $l(z)^{k}$ is $l(-2-n)^{k}$. Thus $l(-2-n)^{k}$ should not appear as a segment of any element in spanning set. Similarly we can list further leading terms in the quotient:
\begin{itemize}
\item Leading term of the coefficient of $z^{nk}$ in $l(z)^{k-1}g(z)$: $$l(-2-n)^{k-1}g(-\frac{3}{2}-n).$$
\item Leading term of the coefficient of $z^{n(k-1-i)+(n-1)i+n}$ in $l(z)^{k-1}g(z)$: $$l(-2-n)^{k-1-i}g(-\frac{3}{2}-n)l(-2-n+1)^{i}\quad  (0\leq i \leq k-1).$$
\end{itemize}
Now we obtain a smaller spanning set where above three type leading terms can not appear inside any ordered monomial. More precisely, any element in this spanning set is of the form $$w(b_{1})w(b_{2})\ldots w(b_{m})$$ where $b_{i}\geq b_{i+1}$,  $w(a)=l(a)$ if $a\in \mathbb{Z}$ and $w(a)=g(a)$ if $a\in \frac{1}{2}+\mathbb{Z}.$ And the fact that $g(a)$ is odd implies that no half-odd-integer is repeated in $\left\{b_{1},b_{2},\ldots,b_{m}\right\}.$   Moreover we have the condition \[b_{j}-b_{j-k+1}>1,\quad \text{if}\quad b_{j}\in \mathbb{Z},\] because \begin{align*}&l(-2-n)^{k},\quad l(-2-n)^{k-1}g(-\frac{3}{2}-n),\\ &l(-2-n)^{k-1-i}g(-\frac{3}{2}-n)l(-2-n+1)^{i} \quad (1\leq i \leq k-1)\end{align*} are leading terms of some elements in the quotient ideal. We also have a condition  \[b_{j}-b_{j-k+2}\geq1 \quad \text{if}\quad  b_{j}\in \frac{1}{2}\mathbb{Z},\] because \[g(-\frac{3}{2}-n)l(-2-n+1)^{k-1}\] is the leading term of some element in the quotient ideal. So we have  $$HS_{q}(J_{\infty}(\mathbb{C}[l,g]/\langle  l^{k},l^{k-1}g\rangle)\leq \sum_{n=0}^{\infty}D_{k,1}(n)q^{\frac{n}{2}}={\rm ch}[gr(L_{c_{2,4k}}^{N=1})](q).$$ Meanwhile the surjectivity of $\psi$ implies that $$HS_{q}(J_{\infty}(\mathbb{C}[l,g]/\langle  l^{k},l^{k-1}g\rangle )\geq {\rm ch}[gr(L_{c_{2,4k}}^{N=1})](q).$$ Thus $HS_{q}(J_{\infty}(\mathbb{C}[l,g]/\langle  l^{k},l^{k-1}g\rangle )= {\rm ch}[gr(L_{c_{2,4k}}^{N=1})](q)$ and $\psi$ is an isomorphism.
It implies that  above spanning set is a basis of the jet superalgebra. The image of basis of jet superalgebra under map $\psi$ is a basis of  $gr(L_{c_{2,4k}}^{N=1}).$

 We have following result which is a super-analog of  \cite[Theorem 16.13]{van2018chiral}:
\begin{thm}
Let $p'>p\geq 2$ satisfy $\frac{p'-p}{2}$ and $p$ are coprime positive integers. We let $L_{c_{p,p'}}^{N=1}$ denote the simple $N=1$ vertex superalgebra associated with $N=1$ superconformal $(p,p')$-minimal model of central charge $c_{p,p'}=\frac{3}{2}(1-\frac{2(p'-p)^{2}}{pp'})$. Then the map $\psi$ is an isomorphism if and only if $(p,p')=(2,4k)$, $(k\geq 1).$
\end{thm}
\begin{proof}
We first consider $C_2$-algebra of $L_{c_{p,p'}}^{N=1}$. We let \[|c_{p,p'}|=\frac{(p-1)(p'-1)}{4}+\frac{1+(-1)^{pp'}}{8}\in \mathbb{N}.\] When $p$ and $p'$ are both even, according to \cite[Section 4]{milas2007characters}, there are two null vectors which survive in $R_{V_{c_{p,p'}}^{N=1}}$, i.e. $L_{(-2)}^{|c_{p,p'}|}\mathbf{1}$ and $L_{(-2)}^{|c_{p,p'}|-1}G_{-\frac{3}{2}}\mathbf{1}$. They generate the quotient ideal of $R_{V_{c_{p,p'}}^{N=1}}$ in the vacuum algebra.  In this case, the $C_{2}$-algebra $R_{L_{c_{p,p'}}^{N=1}}$ is isomorphic to $$\mathbb{C}[l,g]/\langle l^{|c_{p,p'}|}, l^{|c_{p,p'}|-1}g \rangle.$$ When $p$ and $p'$ are both odd, again from \cite[Section 4]{milas2007characters}, the null vector $L_{(-2)}^{|c_{p,p'}|}\bf{1}$ generates the quotient ideal of $R_{L_{c_{p,p'}}^{N=1}}.$ The $C_{2}$-algebra is isomorphic to $$\mathbb{C}[l,g]/\langle l^{|c_{p,p'}|}\rangle.$$

It is clear that $HS_{q}(J_{\infty}(\mathbb{C}[l,g]/\langle l^{|c_{p,p'}|}\rangle )$ does not equal to ${\rm ch}[L_{c_{p,p'}}^{N=1}](q)$ when $p$ and $p'$ are both odd. Thus $\psi$ is not an isomorphism in this case. Let $p$ and $p'$ be both even. Suppose $(p,p')\notin \left\{(2,4k)\;|\;(k\geq 1)\right\}$ and $\psi$ is an isomorphism for $L_{c_{p,p'}}^{N=1}$. Then \[HS_{q}(J_{\infty}(\mathbb{C}[l,g]/\langle l^{|c_{p,p'}|}, l^{|c_{p,p'}|-1}g \rangle )={\rm ch}[L_{c_{p,p'}}^{N=1}](q).\] On the other hand, we have shown that \[HS_{q}(J_{\infty}(\mathbb{C}[l,g]/\langle l^{k},l^{k-1}g\rangle )={\rm ch}[L_{c_{2,4k}}^{N=1}](q),\;(k\geq 1).\] Therefore the character of $L_{c_{p,p'}}^{N=1}$ must coincide with the character of $L_{c_{2,4k}}^{N=1}$ for some $k$. But according to \cite{melzer1994supersymmetric}, the character of $L_{c_{p,p'}}^{N=1}$ is \[{\rm ch}[L_{c_{p,p'}}^{N=1}](q)=\frac{ \prod_{i \geq 1} (1+q^{i-1/2})}{\prod_{i \geq 1}(1-q^i)} \sum_{j \in
\mathbb{Z}}\left(q^{\frac{{j(jpp'+p'-p)}}{2}}-q^{\frac{{(jp+1)(jp'+1)}}{2}} \right),\] and it is easy to verify from the numerator that no two $N=1$ minimal vertex algebras have the same character. This is a contradiction. Thus the statement is proved.

\end{proof}

\section{Extended Virasoro vertex algebras}
 For a simple Virasoro vertex algebra $L_{Vir}(c_{2,2k+1},0)$ coming from $(2,2k+1)$-minimal model, according to \cite{feigin1993coinvariants}, we know that $R_{L_{Vir}(c_{2,2k+1},0)}\cong \mathbb{C}[x]/(x^{k})$ and $\psi$ is an isomorphism. Let $p$ and $p'$ be two positive coprime integers satisfying $p>p'\geq 2$. It is easy to see that $\psi$ is an isomorphism if and only if $(p,p')=(2,2k+1)$ (see \cite[Theorem 16.13]{van2018chiral}). Recently, the authors displayed the kernel of $\psi$ \cite[Theorem 2]{van2020singular} for the $c=\frac12$ Ising model vertex algebra  $L_{Vir}(c_{3,4},0)$,  based on a new fermionic  character formula of $L_{Vir}(c_{3,4},0)$. 

If we consider extended Virasoro vertex algebras associated with minimal model which is not necessarily a $(2,2k+1)$-minimal model, we might still have that $\psi $ is isomorphism. Our discussion is heavily motivated by \cite{jacob2006embedding} where the combinatorics of (super)extensions of $(3,p)$-minimal vertex algebras were discussed.

\begin{ex}
For the free fermion model  $\mathcal{F} =L_{Vir}(c_{(3,4)},0) \oplus L_{Vir}(c_{(3,4)},\frac12) $, $\psi$ is clearly
an isomorphism as discussed in Proposition \ref{iso fermion}.
\end{ex}

\begin{ex}
The $L_{c_{2,8}}^{N=1}$ minimal vertex superalgebra has the following realization: \[L_{c_{(2,8)}}^{N=1}\cong L_{Vir}(c_{(3,8)},0)\oplus L_{Vir}(c_{(3,8)},\frac{3}{2}).\] This realization is called extended algebra and was studied in \cite{jacob2006embedding}. The map $\psi$ is not an isomorphism in the case of $L_{Vir}(c_{(3,8)},0)$. But we have shown that for the extended algebra of $L_{Vir}(c_{(3,8)},0)$, the map $\psi$ is an isomorphism.
This model was analyzed from a different perspective in \cite{hao2020}.
\end{ex}

\begin{ex}
 Next let us consider $V=L_{Vir}(c_{(3,10)},0)\oplus L_{Vir}(c_{(3,10)},2)$. It is well-known that \[ L_{Vir}(c_{(2,5)},0)\otimes L_{Vir}(c_{(2,5)},0)\cong L_{Vir}(c_{(3,10)},0)\oplus L_{Vir}(c_{(3,10)},2). \] We let $\omega_1$ and $\omega_2$ be conformal vectors of the first factor and the second factor of $L_{Vir}(c_{(2,5)},0)\otimes L_{Vir}(c_{(2,5)},0)$. Then the isomorphism map $f$ sends $\omega_{1}+\omega_{2}$ to the conformal vector $\omega$ of $L_{Vir}(c_{(3,10)},0)$ and $\omega_{1}-\omega_{2}$ to the lowest weight vector $\phi$ of $L_{Vir}(c_{(3,10)},2)$. Since we know that \[J_{\infty}(R_{L_{Vir}(c_{(2,5)},0)})\cong J_{\infty}(\mathbb{C}[x]/\langle x^{2}\rangle)\cong gr(L_{Vir}(c_{2,5},0)),\] the map $\psi$ is an isomorphism for $V$, i.e. \[J_{\infty}(R_{V})=J_{\infty}(R_{L_{Vir}(c_{(2,5)},0)}\otimes R_{L_{Vir}(c_{(2,5)},0)})\cong J_{\infty}(\mathbb{C}[x,y]/\langle x^{2},y^{2}\rangle) \cong gr(V).\] For $L_{Vir}(c_{(3,10)},0)\oplus L_{Vir}(c_{(3,10)},2)$, its $C_{2}$-algebra is isomorphic to $$\mathbb{C}[u,v]/\langle uv,u^{2}+v^{2},u^{3},v^{3}\rangle$$ after we identify $x+y$, $x-y$ in $\mathbb{C}[x,y]/\langle x^{2},y^{2}\rangle$ with $u$ and $v$, respectively.
\end{ex}

 \begin{rem}We also know from \cite{jacob2006embedding} that the normalized parafermionic character of $V=L_{Vir}(c_{(3,10)},0)\oplus L_{Vir}(c_{(3,10)},2)$ is given by  \[{\rm ch}[V](q)= \sum_{n_1,n_2,m_1 \geq 0} \frac{q^{(n_1+n_2+m_1)(n_1+n_2)+n_2(n_2+m_1)+m_1^2+m_1 +n_1+2n_2}}{(q)_{n_1}(q)_{n_2} (q)_{m_1}}.\] 
 Next let us consider the jet algebra $$J_{\infty}(\mathbb{C}[u,v]/\langle u^{2},v^{3},uv\rangle$$ where degrees of $u$ and $v$ are both 2. Clearly, it has the following spanning set:
 \[u_{(-n_{1})}\ldots u_{(-n_{N})}v_{(-m_{1})}\ldots v_{(-m_{M})}\] subject to constraints:
\begin{itemize}
    \item[(a)]\it{difference two condition at distance 1:} $n_{i}\geq n_{i+1}+2 $.
    \item[(b)]\it{difference two condition at distance 2:} 
    $m_{i}\geq m_{i+2}+2.$
    \item[(c)]\it{boundary condition:} $n_{N}\geq 2+M$
\end{itemize} where conditions $(a),$ $(b),$ $(c)$ are coming from $(u^{2})_{\partial}$, $(v^{3})_{\partial}$, $(uv)_{\partial}$ in the quotient ideal of the jet algebra. Meanwhile according to Proposition \ref{slk} and Theorem \ref{prin} , we know that \begin{align*} &J_{\infty}(\mathbb{C}[u]/\langle u^2 \rangle)\cong gr(W_{\Lambda_{1,0}}),\\   &J_{\infty}(\mathbb{C}[v]/\langle v^3 \rangle)\cong gr(W_{\Lambda_{2,0}}), \\ & J_{\infty}(\mathbb{C}[u,v]/\langle uv \rangle)\cong gr(W_{\Gamma}), \end{align*} where $\Gamma$ is the graph $\circ-\circ$. Using three realizations of jet algebras and Gordon-Andrews character formulas from \cite{feigin1993quasi,calinescu2008vertex}, it is not hard to see that the above spanning set, subject to constraints $(a)$-$(c)$, would produce a basis of the jet algebra $J_{\infty}(\mathbb{C}[u,v]/\langle u^{2},v^{3},uv\rangle)$ whose Hilbert series is given by  \begin{align*}\sum_{n_1,n_2,m_1 \geq 0} \frac{q^{(n_1+n_2+m_1)(n_1+n_2)+n_2(n_2+m_1)+m_1^2+m_1 +n_1+2n_2}}{(q)_{n_1}(q)_{n_2} (q)_{m_1}}.\end{align*}
One the other hand, the normalized character formula for $V=L_{Vir}(c_{(2,5)},0)\otimes L_{Vir}(c_{(2,5)},0)$ is \[{\rm ch}[V](q)=\sum_{n_1,n_2 \geq 0} \frac{q^{n_1^2+n_2^2+n_1+n_2}}{(q)_{n_1}(q)_{n_2}}.\] Thus we have Hilbert series identities:
 \begin{align*}&HS_{q}(J_{\infty}(\mathbb{C}[x,y]/(x^{2},y^{2}))= HS_q(J_{\infty}(\mathbb{C}[u,v]/\langle uv,u^{2}+v^{2},u^{3},v^{3}\rangle)\\ &=HS_{q}( J_{\infty}(\mathbb{C}[u,v]/\langle u^{2},v^{3},uv)\rangle)\end{align*} and
 \begin{align*}\sum_{n_1,n_2,m_1 \geq 0} \frac{q^{(n_1+n_2+m_1)(n_1+n_2)+n_2(n_2+m_1)+m_1^2+m_1 +n_1+2n_2}}{(q)_{n_1}(q)_{n_2} (q)_{m_1}}=\sum_{n_1,n_2 \geq 0} \frac{q^{n_1^2+n_2^2+n_1+n_2}}{(q)_{n_1}(q)_{n_2}}.\end{align*}
\end{rem}

\bigskip

\section{Conclusion and future work}

In this work, we examined the injectivity of the $\psi$ map for several types of vertex algebras by making use of variety of
methods, including $q$-series identities. More precisely,

\begin{itemize}

\item[(1)] The concept of jet algebra for vertex algebras can be extended to vertex superalgebras extending
ideas of Arakawa in the super setup.
We showed on several familiar examples, e.g. $(2,4k)$ superconformal vertex algebras, that
the $\psi$ map is an isomorphism. However, we also gave counterexamples coming
from other $N=1$ minimal models and a certain odd rank one lattice vertex algebra.

\item[(2)]  We analyzed in great depth principal subspaces of lattice vertex algebras and affine vertex algebras and
showed that the $\psi$-map is isomorphism for many examples.

\item[(3)] We investigated the jet algebras coming from graphs. Interestingly, in some examples their Hilbert series are (mixed) mock modular forms.

\end{itemize}

\begin{rem}

If $L$ is a root lattice of a Lie algebra of type $D$ or $E_6$, $E_7$ and $E_8$, we expect that the $\psi$ map is an isomorphism for the FS principal subspace $W_{L}$. We will address this in our future work \cite{Li.Milas}.

\end{rem}

\begin{rem}

For the simple affine vertex algebra  $L_{\widehat{\mathfrak{g}}}(k,0)$, we know that $\psi$ is an isomorphism if $\mathfrak{g}$ is of type $C_{n}$ $(n\geq 2)$ for $k=1$ and $\mathfrak{g}=sl_{2}$ for any $k \in \mathbb{N}$.
But we expect that $\psi$ is isomorphism for any $\mathfrak{g}$ and any $k \in \mathbb{N}$, and moreover, we expect that
$\psi$ is isomorphism for the FS-principal subspaces thereof.

\end{rem}

\color{black}

\bibliography{reference}
\bibliographystyle{alpha}
\end{document}